\documentclass[conference,a4paper]{IEEEtran}
\IEEEoverridecommandlockouts
\usepackage{cite}
\usepackage{longtable}
\usepackage{graphicx}
\usepackage{multicol}
\usepackage{xcolor,colortbl}
\usepackage{mleftright}
\usepackage{tikz}
\usepackage{pgfplots}

\usepackage{enumerate}
\usepackage{amsmath}
\usepackage{amssymb}
\usepackage{amsfonts}
\usepackage{amsthm}
\usepackage{mathrsfs}
\usepackage{xspace}
\usepackage{bm}
\usepackage{fancyref}
\usepackage{textcomp}
\usepackage[Symbol]{upgreek}

\newcommand{\safemath}[2]{\newcommand{#1}{\ensuremath{#2}\xspace}}

\newcommand{\ssa}{\mathsf{a}}
\newcommand{\ssb}{\mathsf{b}}
\newcommand{\ssc}{\mathsf{c}}
\newcommand{\ssd}{\mathsf{d}}
\newcommand{\sse}{\mathsf{e}}
\newcommand{\ssf}{\mathsf{f}}
\newcommand{\ssg}{\mathsf{g}}
\newcommand{\ssh}{\mathsf{h}}
\newcommand{\ssi}{\mathsf{i}}
\newcommand{\ssj}{\mathsf{j}}
\newcommand{\ssk}{\mathsf{k}}
\newcommand{\ssl}{\mathsf{l}}
\newcommand{\ssm}{\mathsf{m}}
\newcommand{\ssn}{\mathsf{n}}
\newcommand{\sso}{\mathsf{o}}
\newcommand{\ssp}{\mathsf{p}}
\newcommand{\ssq}{\mathsf{q}}
\newcommand{\ssr}{\mathsf{r}}
\newcommand{\sss}{\mathsf{s}}
\newcommand{\sst}{\mathsf{t}}
\newcommand{\ssu}{\mathsf{u}}
\newcommand{\ssv}{\mathsf{v}}
\newcommand{\ssw}{\mathsf{w}}
\newcommand{\ssx}{\mathsf{x}}
\newcommand{\ssy}{\mathsf{y}}
\newcommand{\ssz}{\mathsf{z}}

\safemath{\bmsa}{\bm{\ssa}}
\safemath{\bmsb}{\bm{\ssb}}
\safemath{\bmsc}{\bm{\ssc}}
\safemath{\bmsd}{\bm{\ssd}}
\safemath{\bmse}{\bm{\sse}}
\safemath{\bmsf}{\bm{\ssf}}
\safemath{\bmsg}{\bm{\ssg}}
\safemath{\bmsh}{\bm{\ssh}}
\safemath{\bmsi}{\bm{\ssi}}
\safemath{\bmsj}{\bm{\ssj}}
\safemath{\bmsk}{\bm{\ssk}}
\safemath{\bmsl}{\bm{\ssl}}
\safemath{\bmsm}{\bm{\ssm}}
\safemath{\bmsn}{\bm{\ssn}}
\safemath{\bmso}{\bm{\sso}}
\safemath{\bmsp}{\bm{\ssp}}
\safemath{\bmsq}{\bm{\ssq}}
\safemath{\bmsr}{\bm{\ssr}}
\safemath{\bmss}{\bm{\sss}}
\safemath{\bmst}{\bm{\sst}}
\safemath{\bmsu}{\bm{\ssu}}
\safemath{\bmsv}{\bm{\ssv}}
\safemath{\bmsw}{\bm{\ssw}}
\safemath{\bmsx}{\bm{\ssx}}
\safemath{\bmsy}{\bm{\ssy}}
\safemath{\bmsz}{\bm{\ssz}}

\bmdefine{\bmualphad}{\upalpha}
\bmdefine{\bmubetad}{\upbeta}
\bmdefine{\bmuchid}{\upchi}
\bmdefine{\bmudeltad}{\updelta}
\bmdefine{\bmuepsilond}{\upepsilon}
\bmdefine{\bmuvarepsilond}{\upvarepsilon}
\bmdefine{\bmuetad}{\upeta}
\bmdefine{\bmugammad}{\upgamma}
\bmdefine{\bmuiotad}{\upiota}
\bmdefine{\bmukappad}{\upkappa}
\bmdefine{\bmulambdad}{\uplambda}
\bmdefine{\bmumud}{\upmu}
\bmdefine{\bmunud}{\upnu}
\bmdefine{\bmuomegad}{\upomega}
\bmdefine{\bmuphid}{\upphi}
\bmdefine{\bmuvarphid}{\upvarphi}
\bmdefine{\bmupid}{\uppi}
\bmdefine{\bmuvarpid}{\upvarpi}
\bmdefine{\bmupsid}{\uppsi}
\bmdefine{\bmurhod}{\uprho}
\bmdefine{\bmuvarrhod}{\upvarrho}
\bmdefine{\bmusigmad}{\upsigma}
\bmdefine{\bmuvarsigmad}{\upvarsigma}
\bmdefine{\bmutaud}{\uptau}
\bmdefine{\bmuthetad}{\uptheta}
\bmdefine{\bmuvarthetad}{\upvartheta}
\bmdefine{\bmuupsilond}{\upupsilon}
\bmdefine{\bmuxid}{\upxi}
\bmdefine{\bmuzetad}{\upzeta}

\safemath{\bmua}{\mathbf{a}}
\safemath{\bmub}{\mathbf{b}}
\safemath{\bmuc}{\mathbf{c}}
\safemath{\bmud}{\mathbf{d}}
\safemath{\bmue}{\mathbf{e}}
\safemath{\bmuf}{\mathbf{f}}
\safemath{\bmug}{\mathbf{g}}
\safemath{\bmuh}{\mathbf{h}}
\safemath{\bmui}{\mathbf{i}}
\safemath{\bmuj}{\mathbf{j}}
\safemath{\bmuk}{\mathbf{k}}
\safemath{\bmul}{\mathbf{l}}
\safemath{\bmum}{\mathbf{m}}
\safemath{\bmun}{\mathbf{n}}
\safemath{\bmuo}{\mathbf{o}}
\safemath{\bmup}{\mathbf{p}}
\safemath{\bmuq}{\mathbf{q}}
\safemath{\bmur}{\mathbf{r}}
\safemath{\bmus}{\mathbf{s}}
\safemath{\bmut}{\mathbf{t}}
\safemath{\bmuu}{\mathbf{u}}
\safemath{\bmuv}{\mathbf{v}}
\safemath{\bmuw}{\mathbf{w}}
\safemath{\bmux}{\mathbf{x}}
\safemath{\bmuy}{\mathbf{y}}
\safemath{\bmuz}{\mathbf{z}}

\safemath{\bmualpha}{\bmualphad}
\safemath{\bmubeta}{\bmubetad}
\safemath{\bmuchi}{\bumchid}
\safemath{\bmudelta}{\bmudeltad}
\safemath{\bmuepsilon}{\bmuepsilond}
\safemath{\bmuvarepsilon}{\bmuvarepsilond}
\safemath{\bmueta}{\bmuetad}
\safemath{\bmugamma}{\bmugammad}
\safemath{\bmuiota}{\bmuiotad}
\safemath{\bmukappa}{\bmukappad}
\safemath{\bmulambda}{\bmulambdad}
\safemath{\bmumu}{\bmumud}
\safemath{\bmunu}{\bmunud}
\safemath{\bmuomega}{\bmuomegad}
\safemath{\bmuphi}{\bmuphid}
\safemath{\bmuvarphi}{\bmuvarphid}
\safemath{\bmupi}{\bmupid}
\safemath{\bmuvarpi}{\bmuvarpid}
\safemath{\bmupsi}{\bmupsid}
\safemath{\bmurho}{\bmurhod}
\safemath{\bmuvarrho}{\bmuvarrhod}
\safemath{\bmusigma}{\bmusigmad}
\safemath{\bmuvarsigma}{\bmuvarsigmad}
\safemath{\bmutau}{\bmutaud}
\safemath{\bmutheta}{\bmuthetad}
\safemath{\bmuvartheta}{\bmuvarthetad}
\safemath{\bmuupsilon}{\bmuupsilond}
\safemath{\bmuxi}{\bmuxid}
\safemath{\bmuzeta}{\bmuzetad}

\bmdefine{\bmiad}{a}
\bmdefine{\bmibd}{b}
\bmdefine{\bmicd}{c}
\bmdefine{\bmidd}{d}
\bmdefine{\bmied}{e}
\bmdefine{\bmifd}{f}
\bmdefine{\bmigd}{g}
\bmdefine{\bmihd}{h}
\bmdefine{\bmiid}{i}
\bmdefine{\bmijd}{j}
\bmdefine{\bmikd}{k}
\bmdefine{\bmild}{l}
\bmdefine{\bmimd}{m}
\bmdefine{\bmind}{n}
\bmdefine{\bmiod}{o}
\bmdefine{\bmipd}{p}
\bmdefine{\bmiqd}{q}
\bmdefine{\bmird}{r}
\bmdefine{\bmisd}{s}
\bmdefine{\bmitd}{t}
\bmdefine{\bmiud}{u}
\bmdefine{\bmivd}{v}
\bmdefine{\bmiwd}{w}
\bmdefine{\bmixd}{x}
\bmdefine{\bmiyd}{y}
\bmdefine{\bmizd}{z}

\bmdefine{\bmialphad}{\alpha}
\bmdefine{\bmibetad}{\beta}
\bmdefine{\bmichid}{\chi}
\bmdefine{\bmideltad}{\delta}
\bmdefine{\bmiepsilond}{\epsilon}
\bmdefine{\bmivarepsilond}{\varepsilon}
\bmdefine{\bmietad}{\eta}
\bmdefine{\bmigammad}{\gamma}
\bmdefine{\bmiiotad}{\iota}
\bmdefine{\bmikappad}{\kappa}
\bmdefine{\bmivarkappad}{\varkappa}
\bmdefine{\bmilambdad}{\lambda}
\bmdefine{\bmimud}{\mu}
\bmdefine{\bminud}{\nu}
\bmdefine{\bmiomegad}{\omega}
\bmdefine{\bmiphid}{\phi}
\bmdefine{\bmivarphid}{\varphi}
\bmdefine{\bmipid}{\pi}
\bmdefine{\bmivarpid}{\varpi}
\bmdefine{\bmipsid}{\psi}
\bmdefine{\bmirhod}{\rho}
\bmdefine{\bmivarrhod}{\varrho}
\bmdefine{\bmisigmad}{\sigma}
\bmdefine{\bmivarsigmad}{\varsigma}
\bmdefine{\bmitaud}{\tau}
\bmdefine{\bmithetad}{\theta}
\bmdefine{\bmivarthetad}{\vartheta}
\bmdefine{\bmiupsilond}{\upsilon}
\bmdefine{\bmixid}{\xi}
\bmdefine{\bmizetad}{\zeta}

\safemath{\bmia}{\bmiad}
\safemath{\bmib}{\bmibd}
\safemath{\bmic}{\bmicd}
\safemath{\bmid}{\bmidd}
\safemath{\bmie}{\bmied}
\safemath{\bmif}{\bmifd}
\safemath{\bmig}{\bmigd}
\safemath{\bmih}{\bmihd}
\safemath{\bmii}{\bmiid}
\safemath{\bmij}{\bmijd}
\safemath{\bmik}{\bmikd}
\safemath{\bmil}{\bmild}
\safemath{\bmim}{\bmimd}
\safemath{\bmin}{\bmind}
\safemath{\bmio}{\bmiod}
\safemath{\bmip}{\bmipd}
\safemath{\bmiq}{\bmiqd}
\safemath{\bmir}{\bmird}
\safemath{\bmis}{\bmisd}
\safemath{\bmit}{\bmitd}
\safemath{\bmiu}{\bmiud}
\safemath{\bmiv}{\bmivd}
\safemath{\bmiw}{\bmiwd}
\safemath{\bmix}{\bmixd}
\safemath{\bmiy}{\bmiyd}
\safemath{\bmiz}{\bmizd}

\safemath{\bmialpha}{\bmialphad}
\safemath{\bmibeta}{\bmibetad}
\safemath{\bmichi}{\bmichid}
\safemath{\bmidelta}{\bmideltad}
\safemath{\bmiepsilon}{\bmiepsilond}
\safemath{\bmivarepsilon}{\bmivarepsilond}
\safemath{\bmieta}{\bmietad}
\safemath{\bmigamma}{\bmigammad}
\safemath{\bmiiota}{\bmiiotad}
\safemath{\bmikappa}{\bmikappad}
\safemath{\bmivarkappa}{\bmivarkappad}
\safemath{\bmilambda}{\bmilambdad}
\safemath{\bmimu}{\bmimud}
\safemath{\bminu}{\bminud}
\safemath{\bmiomega}{\bmiomegad}
\safemath{\bmiphi}{\bmiphid}
\safemath{\bmivarphi}{\bmivarphid}
\safemath{\bmipi}{\bmipid}
\safemath{\bmivarpi}{\bmivarpid}
\safemath{\bmipsi}{\bmipsid}
\safemath{\bmirho}{\bmirhod}
\safemath{\bmivarrho}{\bmivarrhod}
\safemath{\bmisigma}{\bmisigmad}
\safemath{\bmivarsigma}{\bmivarsigmad}
\safemath{\bmitau}{\bmitaud}
\safemath{\bmitheta}{\bmithetad}
\safemath{\bmivartheta}{\bmivarthetad}
\safemath{\bmiupsilon}{\bmiupsilond}
\safemath{\bmixi}{\bmixid}
\safemath{\bmizeta}{\bmizetad}

\bmdefine{\bmuDeltad}{\Updelta}
\bmdefine{\bmuGammad}{\Upgamma}
\bmdefine{\bmuLambdad}{\Uplambda}
\bmdefine{\bmuOmegad}{\Upomega}
\bmdefine{\bmuPhid}{\Upphi}
\bmdefine{\bmuPid}{\Uppi}
\bmdefine{\bmuPsid}{\Uppsi}
\bmdefine{\bmuSigmad}{\Upsigma}
\bmdefine{\bmuThetad}{\Uptheta}
\bmdefine{\bmuUpsilond}{\Upupsilon}
\bmdefine{\bmuXid}{\Upxi}

\safemath{\bmuA}{\mathbf{A}}
\safemath{\bmuB}{\mathbf{B}}
\safemath{\bmuC}{\mathbf{C}}
\safemath{\bmuD}{\mathbf{D}}
\safemath{\bmuE}{\mathbf{E}}
\safemath{\bmuF}{\mathbf{F}}
\safemath{\bmuG}{\mathbf{G}}
\safemath{\bmuH}{\mathbf{H}}
\safemath{\bmuI}{\mathbf{I}}
\safemath{\bmuJ}{\mathbf{J}}
\safemath{\bmuK}{\mathbf{K}}
\safemath{\bmuL}{\mathbf{L}}
\safemath{\bmuM}{\mathbf{M}}
\safemath{\bmuN}{\mathbf{N}}
\safemath{\bmuO}{\mathbf{O}}
\safemath{\bmuP}{\mathbf{P}}
\safemath{\bmuQ}{\mathbf{Q}}
\safemath{\bmuR}{\mathbf{R}}
\safemath{\bmuS}{\mathbf{S}}
\safemath{\bmuT}{\mathbf{T}}
\safemath{\bmuU}{\mathbf{U}}
\safemath{\bmuV}{\mathbf{V}}
\safemath{\bmuW}{\mathbf{W}}
\safemath{\bmuX}{\mathbf{X}}
\safemath{\bmuY}{\mathbf{Y}}
\safemath{\bmuZ}{\mathbf{Z}}

\safemath{\bmuZero}{\mathbf{0}}
\safemath{\bmuOne}{\mathbf{1}}

\safemath{\bmuDelta}{\bmuDeltad}
\safemath{\bmuGamma}{\bmuGammad}
\safemath{\bmuLambda}{\bmuLambdad}
\safemath{\bmuOmega}{\bmuOmegad}
\safemath{\bmuPhi}{\bmuPhid}
\safemath{\bmuPi}{\bmuPid}
\safemath{\bmuPsi}{\bmuPsid}
\safemath{\bmuSigma}{\bmuSigmad}
\safemath{\bmuTheta}{\bmuThetad}
\safemath{\bmuUpsilon}{\bmuUpsilond}
\safemath{\bmuXi}{\bmuXid}

\bmdefine{\bmiAd}{A}
\bmdefine{\bmiBd}{B}
\bmdefine{\bmiCd}{C}
\bmdefine{\bmiDd}{D}
\bmdefine{\bmiEd}{E}
\bmdefine{\bmiFd}{F}
\bmdefine{\bmiGd}{G}
\bmdefine{\bmiHd}{H}
\bmdefine{\bmiId}{I}
\bmdefine{\bmiJd}{J}
\bmdefine{\bmiKd}{K}
\bmdefine{\bmiLd}{L}
\bmdefine{\bmiMd}{M}
\bmdefine{\bmiOd}{N}
\bmdefine{\bmiPd}{O}
\bmdefine{\bmiQd}{P}
\bmdefine{\bmiRd}{R}
\bmdefine{\bmiSd}{S}
\bmdefine{\bmiTd}{T}
\bmdefine{\bmiUd}{U}
\bmdefine{\bmiVd}{V}
\bmdefine{\bmiWd}{W}
\bmdefine{\bmiXd}{X}
\bmdefine{\bmiYd}{Y}
\bmdefine{\bmiZd}{Z}

\bmdefine{\bmiDeltad}{\Delta}
\bmdefine{\bmiGammad}{\Gamma}
\bmdefine{\bmiLambdad}{\Lambda}
\bmdefine{\bmiOmegad}{\Omega}
\bmdefine{\bmiPhid}{\Phi}
\bmdefine{\bmiPid}{\Pi}
\bmdefine{\bmiPsid}{\Psi}
\bmdefine{\bmiSigmad}{\Sigma}
\bmdefine{\bmiThetad}{\Theta}
\bmdefine{\bmiUpsilond}{\Upsilon}
\bmdefine{\bmiXid}{\Xi}

\safemath{\bmiA}{\bmiAd}
\safemath{\bmiB}{\bmiBd}
\safemath{\bmiC}{\bmiCd}
\safemath{\bmiD}{\bmiDd}
\safemath{\bmiE}{\bmiEd}
\safemath{\bmiF}{\bmiFd}
\safemath{\bmiG}{\bmiGd}
\safemath{\bmiH}{\bmiHd}
\safemath{\bmiI}{\bmiId}
\safemath{\bmiJ}{\bmiJd}
\safemath{\bmiK}{\bmiKd}
\safemath{\bmiL}{\bmiLd}
\safemath{\bmiM}{\bmiMd}
\safemath{\bmiN}{\bmiNd}
\safemath{\bmiO}{\bmiOd}
\safemath{\bmiP}{\bmiPd}
\safemath{\bmiQ}{\bmiQd}
\safemath{\bmiR}{\bmiRd}
\safemath{\bmiS}{\bmiSd}
\safemath{\bmiT}{\bmiTd}
\safemath{\bmiU}{\bmiUd}
\safemath{\bmiV}{\bmiVd}
\safemath{\bmiW}{\bmiWd}
\safemath{\bmiX}{\bmiXd}
\safemath{\bmiY}{\bmiYd}
\safemath{\bmiZ}{\bmiZd}

\safemath{\bmiDelta}{\bmiDeltad}
\safemath{\bmiGamma}{\bmiGammad}
\safemath{\bmiLambda}{\bmiLambdad}
\safemath{\bmiOmega}{\bmiOmegad}
\safemath{\bmiPhi}{\bmiPhid}
\safemath{\bmiPi}{\bmiPid}
\safemath{\bmiPsi}{\bmiPsid}
\safemath{\bmiSigma}{\bmiSigmad}
\safemath{\bmiTheta}{\bmiThetad}
\safemath{\bmiUpsilon}{\bmiUpsilond}
\safemath{\bmiXi}{\bmiXid}

\safemath{\evA}{\mathcal{A}}
\safemath{\evB}{\mathcal{B}}
\safemath{\evC}{\mathcal{C}}
\safemath{\evD}{\mathcal{D}}
\safemath{\evE}{\mathcal{E}}
\safemath{\evF}{\mathcal{F}}
\safemath{\evG}{\mathcal{G}}
\safemath{\evH}{\mathcal{H}}
\safemath{\evI}{\mathcal{I}}
\safemath{\evJ}{\mathcal{J}}
\safemath{\evK}{\mathcal{K}}
\safemath{\evL}{\mathcal{L}}
\safemath{\evM}{\mathcal{M}}
\safemath{\evN}{\mathcal{N}}
\safemath{\evO}{\mathcal{O}}
\safemath{\evP}{\mathcal{P}}
\safemath{\evQ}{\mathcal{Q}}
\safemath{\evR}{\mathcal{R}}
\safemath{\evS}{\mathcal{S}}
\safemath{\evT}{\mathcal{T}}
\safemath{\evU}{\mathcal{U}}
\safemath{\evV}{\mathcal{V}}
\safemath{\evW}{\mathcal{W}}
\safemath{\evX}{\mathcal{X}}
\safemath{\evY}{\mathcal{Y}}
\safemath{\evZ}{\mathcal{Z}}

\safemath{\setA}{\mathcal{A}}
\safemath{\setB}{\mathcal{B}}
\safemath{\setC}{\mathcal{C}}
\safemath{\setD}{\mathcal{D}}
\safemath{\setE}{\mathcal{E}}
\safemath{\setF}{\mathcal{F}}
\safemath{\setG}{\mathcal{G}}
\safemath{\setH}{\mathcal{H}}
\safemath{\setI}{\mathcal{I}}
\safemath{\setJ}{\mathcal{J}}
\safemath{\setK}{\mathcal{K}}
\safemath{\setL}{\mathcal{L}}
\safemath{\setM}{\mathcal{M}}
\safemath{\setN}{\mathcal{N}}
\safemath{\setO}{\mathcal{O}}
\safemath{\setP}{\mathcal{P}}
\safemath{\setQ}{\mathcal{Q}}
\safemath{\setR}{\mathcal{R}}
\safemath{\setS}{\mathcal{S}}
\safemath{\setT}{\mathcal{T}}
\safemath{\setU}{\mathcal{U}}
\safemath{\setV}{\mathcal{V}}
\safemath{\setW}{\mathcal{W}}
\safemath{\setX}{\mathcal{X}}
\safemath{\setY}{\mathcal{Y}}
\safemath{\setZ}{\mathcal{Z}}
\safemath{\emptySet}{\varnothing}

\safemath{\colA}{\mathscr{A}}
\safemath{\colB}{\mathscr{B}}
\safemath{\colC}{\mathscr{C}}
\safemath{\colD}{\mathscr{D}}
\safemath{\colE}{\mathscr{E}}
\safemath{\colF}{\mathscr{F}}
\safemath{\colG}{\mathscr{G}}
\safemath{\colH}{\mathscr{H}}
\safemath{\colI}{\mathscr{I}}
\safemath{\colJ}{\mathscr{J}}
\safemath{\colK}{\mathscr{K}}
\safemath{\colL}{\mathscr{L}}
\safemath{\colM}{\mathscr{M}}
\safemath{\colN}{\mathscr{N}}
\safemath{\colO}{\mathscr{O}}
\safemath{\colP}{\mathscr{P}}
\safemath{\colQ}{\mathscr{Q}}
\safemath{\colR}{\mathscr{R}}
\safemath{\colS}{\mathscr{S}}
\safemath{\colT}{\mathscr{T}}
\safemath{\colU}{\mathscr{U}}
\safemath{\colV}{\mathscr{V}}
\safemath{\colW}{\mathscr{W}}
\safemath{\colX}{\mathscr{X}}
\safemath{\colY}{\mathscr{Y}}
\safemath{\colZ}{\mathscr{Z}}

\safemath{\opA}{\mathbb{A}}
\safemath{\opB}{\mathbb{B}}
\safemath{\opC}{\mathbb{C}}
\safemath{\opD}{\mathbb{D}}
\safemath{\opE}{\mathbb{E}}
\safemath{\opF}{\mathbb{F}}
\safemath{\opG}{\mathbb{G}}
\safemath{\opH}{\mathbb{H}}
\safemath{\opI}{\mathbb{I}}
\safemath{\opJ}{\mathbb{J}}
\safemath{\opK}{\mathbb{K}}
\safemath{\opL}{\mathbb{L}}
\safemath{\opM}{\mathbb{M}}
\safemath{\opN}{\mathbb{N}}
\safemath{\opO}{\mathbb{O}}
\safemath{\opP}{\mathbb{P}}
\safemath{\opQ}{\mathbb{Q}}
\safemath{\opR}{\mathbb{R}}
\safemath{\opS}{\mathbb{S}}
\safemath{\opT}{\mathbb{T}}
\safemath{\opU}{\mathbb{U}}
\safemath{\opV}{\mathbb{V}}
\safemath{\opW}{\mathbb{W}}
\safemath{\opX}{\mathbb{X}}
\safemath{\opY}{\mathbb{Y}}
\safemath{\opZ}{\mathbb{Z}}
\safemath{\opZero}{\mathbb{O}}
\safemath{\identityop}{\opI}

\safemath{\sca}{a}
\safemath{\scb}{b}
\safemath{\scc}{c}
\safemath{\scd}{d}
\safemath{\sce}{e}
\safemath{\scf}{f}
\safemath{\scg}{g}
\safemath{\sch}{h}
\safemath{\sci}{i}
\safemath{\scj}{j}
\safemath{\sck}{k}
\safemath{\scl}{l}
\safemath{\scm}{m}
\safemath{\scn}{n}
\safemath{\sco}{o}
\safemath{\scp}{p}
\safemath{\scq}{q}
\safemath{\scr}{r}
\safemath{\scs}{s}
\safemath{\sct}{t}
\safemath{\scu}{u}
\safemath{\scv}{v}
\safemath{\scw}{w}
\safemath{\scx}{x}
\safemath{\scy}{y}
\safemath{\scz}{z}

\safemath{\scA}{A}
\safemath{\scB}{B}
\safemath{\scC}{C}
\safemath{\scD}{D}
\safemath{\scE}{E}
\safemath{\scF}{F}
\safemath{\scG}{G}
\safemath{\scH}{H}
\safemath{\scI}{I}
\safemath{\scJ}{J}
\safemath{\scK}{K}
\safemath{\scL}{L}
\safemath{\scM}{M}
\safemath{\scN}{N}
\safemath{\scO}{O}
\safemath{\scP}{P}
\safemath{\scQ}{Q}
\safemath{\scR}{R}
\safemath{\scS}{S}
\safemath{\scT}{T}
\safemath{\scU}{U}
\safemath{\scV}{V}
\safemath{\scW}{W}
\safemath{\scX}{X}
\safemath{\scY}{Y}
\safemath{\scZ}{Z}

\safemath{\scalpha}{\alpha}
\safemath{\scbeta}{\beta}
\safemath{\scchi}{\chi}
\safemath{\scdelta}{\delta}
\safemath{\scepsilon}{\epsilon}
\safemath{\scvarepsilon}{\varepsilon}
\safemath{\sceta}{\eta}
\safemath{\scgamma}{\gamma}
\safemath{\sciota}{\iota}
\safemath{\sckappa}{\kappa}
\safemath{\scvarkappa}{\varkappa}
\safemath{\sclambda}{\lambda}
\safemath{\scmu}{\mu}
\safemath{\scnu}{\nu}
\safemath{\scomega}{\omega}
\safemath{\scphi}{\phi}
\safemath{\scvarphi}{\varphi}
\safemath{\scpi}{\pi}
\safemath{\scvarpi}{\varpi}
\safemath{\scpsi}{\psi}
\safemath{\scrho}{\rho}
\safemath{\scvarrho}{\varrho}
\safemath{\scsigma}{\sigma}
\safemath{\scvarsigma}{\varsigma}
\safemath{\sctau}{\tau}
\safemath{\sctheta}{\theta}
\safemath{\scvartheta}{\vartheta}
\safemath{\scupsilon}{\upsilon}
\safemath{\scxi}{\xi}
\safemath{\sczeta}{\zeta}

\safemath{\veca}{\mathrm{a}}
\safemath{\vecb}{\mathrm{b}}
\safemath{\vecc}{\mathrm{c}}
\safemath{\vecd}{\mathrm{d}}
\safemath{\vece}{\mathrm{e}}
\safemath{\vecf}{\mathrm{f}}
\safemath{\vecg}{\mathrm{g}}
\safemath{\vech}{\mathrm{h}}
\safemath{\veci}{\mathrm{i}}
\safemath{\vecj}{\mathrm{j}}
\safemath{\veck}{\mathrm{k}}
\safemath{\vecl}{\mathrm{l}}
\safemath{\vecm}{\mathrm{m}}
\safemath{\vecn}{\mathrm{n}}
\safemath{\veco}{\mathrm{o}}
\safemath{\vecp}{\mathrm{p}}
\safemath{\vecq}{\mathrm{q}}
\safemath{\vecr}{\mathrm{r}}
\safemath{\vecs}{\mathrm{s}}
\safemath{\vect}{\mathrm{t}}
\safemath{\vecu}{\mathrm{u}}
\safemath{\vecv}{\mathrm{v}}
\safemath{\vecw}{\mathrm{w}}
\safemath{\vecx}{\mathrm{x}}
\safemath{\vecy}{\mathrm{y}}
\safemath{\vecz}{\mathrm{z}}
\safemath{\veczero}{\mathrm{0}}
\safemath{\vecone}{\mathrm{1}}

\safemath{\vecalpha}{\upalpha}
\safemath{\vecbeta}{\upbeta}
\safemath{\vecchi}{\upchi}
\safemath{\vecdelta}{\updelta}
\safemath{\vecepsilon}{\upepsilon}
\safemath{\vecvarepsilon}{\upvarepsilon}
\safemath{\veceta}{\upeta}
\safemath{\vecgamma}{\upgamma}
\safemath{\veciota}{\upiota}
\safemath{\veckappa}{\upkappa}
\safemath{\veclambda}{\uplambda}
\safemath{\vecmu}{\text{\textmu}}
\safemath{\vecnu}{\upnu}
\safemath{\vecomega}{\upomega}
\safemath{\vecphi}{\upphi}
\safemath{\vecvarphi}{\upvarphi}
\safemath{\vecpi}{\uppi}
\safemath{\vecvarpi}{\upvarpi}
\safemath{\vecpsi}{\uppsi}
\safemath{\vecrho}{\uprho}
\safemath{\vecvarrho}{\upvarrho}
\safemath{\vecsigma}{\upsigma}
\safemath{\vecvarsigma}{\upvarsigma}
\safemath{\vectau}{\uptau}
\safemath{\vectheta}{\uptheta}
\safemath{\vecvartheta}{\upvartheta}
\safemath{\vecupsilon}{\upupsilon}
\safemath{\vecxi}{\upxi}
\safemath{\veczeta}{\upzeta}

\safemath{\vecac}{a}
\safemath{\vecbc}{b}
\safemath{\veccc}{c}
\safemath{\vecdc}{d}
\safemath{\vecec}{e}
\safemath{\vecfc}{f}
\safemath{\vecgc}{g}
\safemath{\vechc}{h}
\safemath{\vecic}{i}
\safemath{\vecjc}{j}
\safemath{\veckc}{k}
\safemath{\veclc}{l}
\safemath{\vecmc}{m}
\safemath{\vecnc}{n}
\safemath{\vecoc}{o}
\safemath{\vecpc}{p}
\safemath{\vecqc}{q}
\safemath{\vecrc}{r}
\safemath{\vecsc}{s}
\safemath{\vectc}{t}
\safemath{\vecuc}{u}
\safemath{\vecvc}{v}
\safemath{\vecwc}{w}
\safemath{\vecxc}{x}
\safemath{\vecyc}{y}
\safemath{\veczc}{z}

\safemath{\matA}{\mathrm{A}}
\safemath{\matB}{\mathrm{B}}
\safemath{\matC}{\mathrm{C}}
\safemath{\matD}{\mathrm{D}}
\safemath{\matE}{\mathrm{E}}
\safemath{\matF}{\mathrm{F}}
\safemath{\matG}{\mathrm{G}}
\safemath{\matH}{\mathrm{H}}
\safemath{\matI}{\mathrm{I}}
\safemath{\matJ}{\mathrm{J}}
\safemath{\matK}{\mathrm{K}}
\safemath{\matL}{\mathrm{L}}
\safemath{\matM}{\mathrm{M}}
\safemath{\matN}{\mathrm{N}}
\safemath{\matO}{\mathrm{O}}
\safemath{\matP}{\mathrm{P}}
\safemath{\matQ}{\mathrm{Q}}
\safemath{\matR}{\mathrm{R}}
\safemath{\matS}{\mathrm{S}}
\safemath{\matT}{\mathrm{T}}
\safemath{\matU}{\mathrm{U}}
\safemath{\matV}{\mathrm{V}}
\safemath{\matW}{\mathrm{W}}
\safemath{\matX}{\mathrm{X}}
\safemath{\matY}{\mathrm{Y}}
\safemath{\matZ}{\mathrm{Z}}
\safemath{\matzero}{\mathrm{0}}

\safemath{\matDelta}{\Updelta}
\safemath{\matGamma}{\Upgammma}
\safemath{\matLambda}{\Uplambda}
\safemath{\matOmega}{\Upomega}
\safemath{\matPhi}{\Upphi}
\safemath{\matPi}{\Uppi}
\safemath{\matPsi}{\Uppsi}
\safemath{\matSigma}{\Upsigma}
\safemath{\matTheta}{\Uptheta}
\safemath{\matUpsilon}{\Upupsilon}
\safemath{\matXi}{\Upxi}

\safemath{\matidentity}{\matI}
\safemath{\vecunit}{\vece} 
\safemath{\matone}{\matO}

\safemath{\matAc}{a}
\safemath{\matBc}{b}
\safemath{\matCc}{c}
\safemath{\matDc}{d}
\safemath{\matEc}{e}
\safemath{\matFc}{f}
\safemath{\matGc}{g}
\safemath{\matHc}{h}
\safemath{\matIc}{i}
\safemath{\matJc}{j}
\safemath{\matKc}{k}
\safemath{\matLc}{l}
\safemath{\matMc}{m}
\safemath{\matNc}{n}
\safemath{\matOc}{o}
\safemath{\matPc}{p}
\safemath{\matQc}{q}
\safemath{\matRc}{r}
\safemath{\matSc}{s}
\safemath{\matTc}{t}
\safemath{\matUc}{u}
\safemath{\matVc}{v}
\safemath{\matWc}{w}
\safemath{\matXc}{x}
\safemath{\matYc}{y}
\safemath{\matZc}{z}

\safemath{\rnda}{\mathsf{a}}
\safemath{\rndb}{\mathsf{b}}
\safemath{\rndc}{\mathsf{c}}
\safemath{\rndd}{\mathsf{d}}
\safemath{\rnde}{\mathsf{e}}
\safemath{\rndf}{\mathsf{f}}
\safemath{\rndg}{\mathsf{g}}
\safemath{\rndh}{\mathsf{h}}
\safemath{\rndi}{\mathsf{i}}
\safemath{\rndj}{\mathsf{j}}
\safemath{\rndk}{\mathsf{k}}
\safemath{\rndl}{\mathsf{l}}
\safemath{\rndm}{\mathsf{m}}
\safemath{\rndn}{\mathsf{n}}
\safemath{\rndo}{\mathsf{o}}
\safemath{\rndp}{\mathsf{p}}
\safemath{\rndq}{\mathsf{q}}
\safemath{\rndr}{\mathsf{r}}
\safemath{\rnds}{\mathsf{s}}
\safemath{\rndt}{\mathsf{t}}
\safemath{\rndu}{\mathsf{u}}
\safemath{\rndv}{\mathsf{v}}
\safemath{\rndw}{\mathsf{w}}
\safemath{\rndx}{\mathsf{x}}
\safemath{\rndy}{\mathsf{y}}
\safemath{\rndz}{\mathsf{z}}

\safemath{\rndA}{\bmiA}
\safemath{\rndB}{\bmiB}
\safemath{\rndC}{\bmiC}
\safemath{\rndD}{\bmiD}
\safemath{\rndE}{\bmiE}
\safemath{\rndF}{\bmiF}
\safemath{\rndG}{\bmiG}
\safemath{\rndH}{\bmiH}
\safemath{\rndI}{\bmiI}
\safemath{\rndJ}{\bmiJ}
\safemath{\rndK}{\bmiK}
\safemath{\rndL}{\bmiL}
\safemath{\rndM}{\bmiM}
\safemath{\rndN}{\bmiN}
\safemath{\rndO}{\bmiO}
\safemath{\rndP}{\bmiP}
\safemath{\rndQ}{\bmiQ}
\safemath{\rndR}{\bmiR}
\safemath{\rndS}{\bmiS}
\safemath{\rndT}{\bmiT}
\safemath{\rndU}{\bmiU}
\safemath{\rndV}{\bmiV}
\safemath{\rndW}{\bmiW}
\safemath{\rndX}{\bmiX}
\safemath{\rndY}{\bmiY}
\safemath{\rndZ}{\bmiZ}

\safemath{\rndalpha}{\bmialpha}
\safemath{\rndbeta}{\bmibeta}
\safemath{\rndchi}{\bmichi}
\safemath{\rnddelta}{\bmidelta}
\safemath{\rndepsilon}{\bmiepsilon}
\safemath{\rndvarepsilon}{\bmivarepsilon}
\safemath{\rndeta}{\bmieta}
\safemath{\rndgamma}{\bmigamma}
\safemath{\rndiota}{\bmiiota}
\safemath{\rndkappa}{\bmikappa}
\safemath{\rndlambda}{\bmilambda}
\safemath{\rndmu}{\bmimu}
\safemath{\rndnu}{\bminu}
\safemath{\rndomega}{\bmiomega}
\safemath{\rndphi}{\bmiphi}
\safemath{\rndvarphi}{\bmivarphi}
\safemath{\rndpi}{\bmipi}
\safemath{\rndvarpi}{\bmivarpi}
\safemath{\rndpsi}{\bmipsi}
\safemath{\rndrho}{\bmirho}
\safemath{\rndvarrho}{\bmivarrho}
\safemath{\rndsigma}{\bmisigma}
\safemath{\rndvarsigma}{\bmivarsigma}
\safemath{\rndtau}{\bmitau}
\safemath{\rndtheta}{\bmitheta}
\safemath{\rndvartheta}{\bmivartheta}
\safemath{\rndupsilon}{\bmiupsilon}
\safemath{\rndxi}{\bmixi}
\safemath{\rndzeta}{\bmizeta}

\safemath{\rveca}{\mathbf{a}}
\safemath{\rvecb}{\mathbf{b}}
\safemath{\rvecc}{\mathbf{c}}
\safemath{\rvecd}{\mathbf{d}}
\safemath{\rvece}{\mathbf{e}}
\safemath{\rvecf}{\mathbf{f}}
\safemath{\rvecg}{\mathbf{g}}
\safemath{\rvech}{\mathbf{h}}
\safemath{\rveci}{\mathbf{i}}
\safemath{\rvecj}{\mathbf{j}}
\safemath{\rveck}{\mathbf{k}}
\safemath{\rvecl}{\mathbf{l}}
\safemath{\rvecm}{\mathbf{m}}
\safemath{\rvecn}{\mathbf{n}}
\safemath{\rveco}{\mathbf{o}}
\safemath{\rvecp}{\mathbf{p}}
\safemath{\rvecq}{\mathbf{q}}
\safemath{\rvecr}{\mathbf{r}}
\safemath{\rvecs}{\mathbf{s}}
\safemath{\rvect}{\mathbf{t}}
\safemath{\rvecu}{\mathbf{u}}
\safemath{\rvecv}{\mathbf{v}}
\safemath{\rvecw}{\mathbf{w}}
\safemath{\rvecx}{\mathbf{x}}
\safemath{\rvecy}{\mathbf{y}}
\safemath{\rvecz}{\mathbf{z}}

\safemath{\rvecalpha}{\bmualpha}
\safemath{\rvecbeta}{\bmubeta}
\safemath{\rvecchi}{\bmuchi}
\safemath{\rvecdelta}{\bmudelta}
\safemath{\rvecepsilon}{\bmuepsilon}
\safemath{\rvecvarepsilon}{\bmuvarepsilon}
\safemath{\rveceta}{\bmueta}
\safemath{\rvecgamma}{\bmugamma}
\safemath{\rveciota}{\bmuiota}
\safemath{\rveckappa}{\bmukappa}
\safemath{\rveclambda}{\bmulambda}
\safemath{\rvecmu}{\bmumu}
\safemath{\rvecnu}{\bmunu}
\safemath{\rvecomega}{\bmuomega}
\safemath{\rvecphi}{\bmuphi}
\safemath{\rvecvarphi}{\bmuvarphi}
\safemath{\rvecpi}{\bmupi}
\safemath{\rvecvarpi}{\bmuvarpi}
\safemath{\rvecpsi}{\bmupsi}
\safemath{\rvecrho}{\bmurho}
\safemath{\rvecvarrho}{\bmuvarrho}
\safemath{\rvecsigma}{\bmusigma}
\safemath{\rvecvarsigma}{\bmuvarsigma}
\safemath{\rvectau}{\bmutau}
\safemath{\rvectheta}{\bmutheta}
\safemath{\rvecvartheta}{\bmuvartheta}
\safemath{\rvecupsilon}{\bmuupsilon}
\safemath{\rvecxi}{\bmuxi}
\safemath{\rveczeta}{\bmuzeta}

\safemath{\rvecac}{\rnda}
\safemath{\rvecbc}{\rndb}
\safemath{\rveccc}{\rndc}
\safemath{\rvecdc}{\rndd}
\safemath{\rvecec}{\rnde}
\safemath{\rvecfc}{\rndf}
\safemath{\rvecgc}{\rndg}
\safemath{\rvechc}{\rndh}
\safemath{\rvecic}{\rndi}
\safemath{\rvecjc}{\rndj}
\safemath{\rveckc}{\rndk}
\safemath{\rveclc}{\rndl}
\safemath{\rvecmc}{\rndm}
\safemath{\rvecnc}{\rndn}
\safemath{\rvecoc}{\rndo}
\safemath{\rvecpc}{\rndp}
\safemath{\rvecqc}{\rndq}
\safemath{\rvecrc}{\rndr}
\safemath{\rvecsc}{\rnds}
\safemath{\rvectc}{\rndt}
\safemath{\rvecuc}{\rndu}
\safemath{\rvecvc}{\rndv}
\safemath{\rvecwc}{\rndw}
\safemath{\rvecxc}{\rndx}
\safemath{\rvecyc}{\rndy}
\safemath{\rveczc}{\rndz}

\safemath{\rmatA}{\mathbf{A}}
\safemath{\rmatB}{\mathbf{B}}
\safemath{\rmatC}{\mathbf{C}}
\safemath{\rmatD}{\mathbf{D}}
\safemath{\rmatE}{\mathbf{E}}
\safemath{\rmatF}{\mathbf{F}}
\safemath{\rmatG}{\mathbf{G}}
\safemath{\rmatH}{\mathbf{H}}
\safemath{\rmatI}{\mathbf{I}}
\safemath{\rmatJ}{\mathbf{J}}
\safemath{\rmatK}{\mathbf{K}}
\safemath{\rmatL}{\mathbf{L}}
\safemath{\rmatM}{\mathbf{M}}
\safemath{\rmatN}{\mathbf{N}}
\safemath{\rmatO}{\mathbf{O}}
\safemath{\rmatP}{\mathbf{P}}
\safemath{\rmatQ}{\mathbf{Q}}
\safemath{\rmatR}{\mathbf{R}}
\safemath{\rmatS}{\mathbf{S}}
\safemath{\rmatT}{\mathbf{T}}
\safemath{\rmatU}{\mathbf{U}}
\safemath{\rmatV}{\mathbf{V}}
\safemath{\rmatW}{\mathbf{W}}
\safemath{\rmatX}{\mathbf{X}}
\safemath{\rmatY}{\mathbf{Y}}
\safemath{\rmatZ}{\mathbf{Z}}
\safemath{\rmatDelta}{\bmuDelta}
\safemath{\rmatGamma}{\bmuGamma}
\safemath{\rmatLambda}{\bmuLambda}
\safemath{\rmatOmega}{\bmuOmega}
\safemath{\rmatPhi}{\bmuPhi}
\safemath{\rmatPi}{\bmuPi}
\safemath{\rmatPsi}{\bmuPsi}
\safemath{\rmatSigma}{\bmuSigma}
\safemath{\rmatTheta}{\bmuTheta}
\safemath{\rmatUpsilon}{\bmuUpsilon}
\safemath{\rmatXi}{\bmuXi}

\safemath{\rmatAc}{\rnda}
\safemath{\rmatBc}{\rndb}
\safemath{\rmatCc}{\rndc}
\safemath{\rmatDc}{\rndd}
\safemath{\rmatEc}{\rnde}
\safemath{\rmatFc}{\rndf}
\safemath{\rmatGc}{\rndg}
\safemath{\rmatHc}{\rndh}
\safemath{\rmatIc}{\rndi}
\safemath{\rmatJc}{\rndj}
\safemath{\rmatKc}{\rndk}
\safemath{\rmatLc}{\rndl}
\safemath{\rmatMc}{\rndm}
\safemath{\rmatNc}{\rndn}
\safemath{\rmatOc}{\rndo}
\safemath{\rmatPc}{\rndp}
\safemath{\rmatQc}{\rndq}
\safemath{\rmatRc}{\rndr}
\safemath{\rmatSc}{\rnds}
\safemath{\rmatTc}{\rndt}
\safemath{\rmatUc}{\rndu}
\safemath{\rmatVc}{\rndv}
\safemath{\rmatWc}{\rndw}
\safemath{\rmatXc}{\rndx}
\safemath{\rmatYc}{\rndy}
\safemath{\rmatZc}{\rndz}

\newenvironment{textbmatrix}{	\setlength{\arraycolsep}{2.5pt}%
								\big[\begin{matrix}}{\end{matrix}\big]%
								\raisebox{0.08ex}{\vphantom{M}}}

 \def\btm{\begin{textbmatrix}}
 \def\etm{\end{textbmatrix}}

\DeclareMathOperator{\tr}{tr}				
\DeclareMathOperator{\diag}{diag}			
\DeclareMathOperator{\rank}{rank}			
\DeclareMathOperator{\adj}{adj}				

\safemath{\fun}{\scf}						

\safemath{\vrbl}{x}						
\safemath{\altvrbl}{y}						
\safemath{\aaltvrbl}{z}						

\safemath{\vvrbl}{\vecx}						
\safemath{\altvvrbl}{\vecy}						
\safemath{\aaltvvrbl}{\vecz}						

\safemath{\altfun}{\scg}
\safemath{\aaltfun}{\sch}
\safemath{\bel}{\sce}					
\safemath{\altbel}{\sce}					
\safemath{\frel}{g}					
\safemath{\altfrel}{g}					
\safemath{\dfrel}{\tilde{g}}					
\safemath{\altdfrel}{\tilde{g}}					

\safemath{\mat}{\matA}						
\safemath{\matc}{\matAc}						
\safemath{\altmat}{\matB}						
\safemath{\altmatc}{\matBc}						

\safemath{\vectr}{\vecu}						
\safemath{\vectrc}{\vecuc}						
\safemath{\altvectr}{\vecv}						
\safemath{\aaltvectr}{\vect}						
\safemath{\altvectrc}{\vecvc}						
\safemath{\genvar}{u}						
\safemath{\altgenvar}{v}						

\safemath{\rvectr}{\rvecu}						
\safemath{\rvectrc}{\rvecuc}						
\safemath{\raltvectr}{\rvecv}						
\safemath{\raaltvectr}{\rvect}						
\safemath{\raltvectrc}{\rvecvc}						
\safemath{\rgenvar}{\rndu}						
\safemath{\raltgenvar}{\rndv}						

\newcommand{\nullspace}{\setN}	 			

\newcommand{\ind}[1]{\chi_{#1}}				
\newcommand{\conj}[1]{\ensuremath{#1^{*}}} 	
\newcommand{\tp}[1]{\ensuremath{#1^{\mathsf{T}}}} 		

\newcommand{\inv}[1]{\ensuremath{#1^{-1}}} 	

\safemath{\dirac}{\delta}					
\safemath{\diracp}{\dirac(\time)}			
\safemath{\krond}{\dirac}					
\safemath{\indfun}{I}						
\safemath{\stepfun}{u}						

\safemath{\upto}{\uparrow}
\safemath{\downto}{\downarrow}
\safemath{\iu}{\mathrm{i}}							
\safemath{\maj}{\succ}
\newcommand{\dftmat}[1]{\matF_{#1}}			
\safemath{\mdft}{\dftmat{}}					
\safemath{\runity}{\beta}					
\safemath{\eval}{\lambda}					

\safemath{\veval}{\veclambda}				
\safemath{\littleo}{\sco}					

\let\im\undefined
\safemath{\re}{\Re}				
\safemath{\im}{\Im}				

\safemath{\euclidspace}{\complexset}			
\safemath{\confunspace}{\setC}				
\newcommand{\banachseqspace}[1]{l^{#1}}		
\safemath{\hilseqspace}{\banachseqspace{2}}	
\newcommand{\banachfunspace}[1]{\setL^{#1}}	
\safemath{\hilfunspace}{\banachfunspace{2}}	
\safemath{\hilfunspacep}{\hilfunspace(\complexset)}	
\safemath{\schwarzspace}{\setS}				

\newcommand{\hadj}[1]{#1^{\star}}			

\safemath{\SNR}{\rho} 				
\safemath{\SINR}{\text{\sc sinr}} 				
\safemath{\No}{N_0}							
\safemath{\Es}{E_s}							
\safemath{\Eb}{E_b}							
\safemath{\EbNo}{\frac{\Eb}{\No}}
\safemath{\EsNo}{\frac{\Es}{\No}}
\safemath{\NoVar}{\variance}                 

\let\time\undefined
\safemath{\time}{\sct}						
\safemath{\dtime}{\sck}						
\safemath{\delay}{\sctau}					
\safemath{\ddelay}{\scl}						
\safemath{\doppler}{\scnu}					
\safemath{\ddoppler}{\scm}					
\safemath{\freq}{\scf}						
\safemath{\dfreq}{\scn}						
\safemath{\Dtime}{\Delta\time}
\safemath{\Dfreq}{\Delta\freq}
\safemath{\Ddtime}{\dtime}
\safemath{\Ddfreq}{\dfreq}
\safemath{\bandwidth}{\scB}
\safemath{\maxdoppler}{\doppler_{0}}			
\safemath{\maxdelay}{\delay_{0}}				
\safemath{\spread}{\Delta_{\CHop}}			
\DeclareMathOperator{\CHop}{\ensuremath{\opH}} 
\safemath{\kernel}{\rndk_{\CHop}}			
\safemath{\kernelp}{\kernel(\time,\time')}	
\safemath{\tvir}{\rndh_{\CHop}}				
\safemath{\tvirp}{\tvir(\time,\delay)}		
\safemath{\tvirc}{\conj{\rndh}_{\CHop}}		
\safemath{\tvtf}{\rndl_{\CHop}}				
\safemath{\tvtfp}{\tvtf(\time,\freq)}			
\safemath{\tvtfc}{\conj{\rndl}_{\CHop}}		
\safemath{\spf}{\rnds_{\CHop}}				
\safemath{\spfp}{\spf(\doppler,\delay)}		
\safemath{\spfc}{\conj{\rnds}_{\CHop}}		
\safemath{\bff}{\rndb_{\CHop}}				
\safemath{\bffp}{\bff(\doppler,\freq)}		

\safemath{\irc}{\scr_{\rndh}}				
\safemath{\tfc}{\scr_{\rndl}}				
\safemath{\spc}{\scr_{\rnds}}				
\safemath{\bfc}{\scr_{\rndb}}				

\safemath{\scaf}{\scc_{\rnds}}				
\safemath{\scafp}{\scaf(\doppler,\delay)}		
\safemath{\ccf}{\scc_{\rndl}}				
\safemath{\ccfp}{\ccf(\Dtime,\Dfreq)}			
\safemath{\cic}{\scc_{\rndh}}				
\safemath{\cicp}{\cic(\Dtime,\delay)}			

\safemath{\mi}{I}							
\safemath{\capacity}{C}					

\DeclareMathOperator{\Prob}{\opP}		

\safemath{\normal}{\mathcal{N}}			
\safemath{\jpg}{\mathcal{CN}}			
\safemath{\uniform}{\mathcal{U}}				
\safemath{\mchain}{\leftrightarrow}		

\safemath{\dB}{\,\mathrm{dB}}
\safemath{\dBm}{\,\mathrm{dBm}}
\safemath{\Hz}{\,\mathrm{Hz}}
\safemath{\kHz}{\,\mathrm{kHz}}
\safemath{\MHz}{\,\mathrm{MHz}}
\safemath{\GHz}{\,\mathrm{GHz}}
\safemath{\s}{\,\mathrm{s}}
\safemath{\ms}{\,\mathrm{ms}}
\safemath{\mus}{\,\mathrm{\text{\textmu}s}}
\safemath{\ns}{\,\mathrm{ns}}
\safemath{\ps}{\,\mathrm{ps}}
\safemath{\meter}{\,\mathrm{m}}
\safemath{\mm}{\,\mathrm{mm}}
\safemath{\cm}{\,\mathrm{cm}}
\safemath{\m}{\,\mathrm{m}}
\safemath{\W}{\,\mathrm{W}}
\safemath{\mW}{\, \mathrm{mW}}
\safemath{\J}{\,\mathrm{J}}
\safemath{\K}{\,\mathrm{K}}
\safemath{\bit}{\,\mathrm{bit}}
\safemath{\nat}{\,\mathrm{nat}}

\safemath{\define}{\triangleq}					

\safemath{\equivalent}{\sim}
\safemath{\distas}{\sim}					
\safemath{\sdiff}{\Delta}				
\safemath{\setdiff}{\setminus}				

\safemath{\reals}{\mathbb R}
\safemath{\positivereals}{\reals^{+}}
\safemath{\integers}{\mathbb Z}
\safemath{\posint}{\integers^{+}}
\safemath{\naturals}{\mathbb N}
\safemath{\posnaturals}{\naturals^{+}}
\safemath{\complexset}{\mathbb C}
\safemath{\rationals}{\mathbb Q}

\safemath{\iSet}{\setI}
\safemath{\rel}{\bowtie}					
\safemath{\eqrel}{\sim}					
\safemath{\rlord}{\leq}					
\safemath{\slord}{<}						
\safemath{\rpord}{\preceq}				
\safemath{\rrpord}{\succeq}				
\safemath{\spord}{\prec}					
\safemath{\sig}{\sigma}					
\safemath{\metric}{d}					

\safemath{\setfun}{\Phi}					
\safemath{\measure}{\mu}					
\safemath{\altmeasure}{\lambda}					
\newcommand{\outerm}[1]{#1^{\star}}		
\newcommand{\innerm}[1]{#1_{\star}}		
\safemath{\omeasure}{\outerm{\measure}}		
\safemath{\imeasure}{\innerm{\measure}}		
\safemath{\aecol}{\colS^{\star}_{\measure}} 
\safemath{\emeasure}{\bar{\measure}_{0}}	
\safemath{\rmeasure}{\tilde{\measure}}	
\safemath{\bmeasure}{\measure_{0}}		
\safemath{\glength}{\measure_{\altfun}}	
\safemath{\lebmea}{\lambda}				
\safemath{\blebmea}{\lebmea_{0}}			
\safemath{\sfun}{s}						
\safemath{\absintspace}{\colL^{1}}		
\safemath{\sqintspace}{\colL^{2}}		
\safemath{\abssumspace}{l^{1}}		
\safemath{\sqsumspace}{l^{2}}		

\safemath{\sfield}{\setF}				
\safemath{\vectors}{\setV}				
\safemath{\vecspace}{(\vectors,\sfield)}	
\safemath{\linspace}{\setV}				
\safemath{\clinspace}{(\linspace,\sfield)} 
\safemath{\nspace}{\setU}				
\safemath{\metspace}{\setM}				
\safemath{\bspace}{\setB}				
\safemath{\ipspace}{\setG}				
\safemath{\hilspace}{\setH}				
\safemath{\blospace}{\setG}				
\safemath{\lop}{\opT}					
\safemath{\altlop}{\opS}					
\safemath{\nullsp}{\nullspace(\lop)}		
\safemath{\lfun}{l}						
\safemath{\altlfun}{g}					
\newcommand{\dual}[1]{#1^{'}}			
\safemath{\dsum}{\oplus}					
\safemath{\funspace}{\colL}				
\renewcommand{\adj}[1]{#1^{\times}}		
\safemath{\adjlop}{\adj{\lop}}			
\safemath{\hadjlop}{\hadj{\lop}}			
\safemath{\tow}{\xrightarrow{w}}			
\safemath{\tows}{\xrightarrow{w^{*}}}		
\safemath{\cparam}{\lambda}
\safemath{\lopl}{\lop_{\cparam}}		
\safemath{\iop}{\opI}					
\safemath{\resolop}{\opR}				
\safemath{\resolvent}{\resolop_{\cparam}(\lop)}	
\safemath{\reset}{\setQ}
\safemath{\spectrum}{\setS}
\safemath{\resolset}{\reset(\lop)}		
\safemath{\lopspec}{\spectrum(\lop)}		
\safemath{\pspec}{\spectrum_{p}(\lop)}	
\safemath{\cspec}{\spectrum_{c}(\lop)}	
\safemath{\rspec}{\spectrum_{r}(\lop)}	
\safemath{\ev}{\cparam}					
\newcommand{\specrad}[1]{r_{#1}}			
\safemath{\lopsrad}{\specrad{\lop}}		
\safemath{\pop}{\opP}					

\safemath{\specfam}{\colE}				
\safemath{\specop}{\opE_{\cparam}}		
\safemath{\altspecop}{\opE_{\mu}}		
\safemath{\vmulti}{\vecone}				
\safemath{\unitaryop}{\opU}				
\safemath{\sval}{\sigma}					
\safemath{\corrcoef}{\rho}				
\safemath{\sangle}{\theta}				

\let\time\undefined
\safemath{\iset}{\setI}				
\safemath{\shift}{\nu}
\safemath{\scale}{\alpha}
\safemath{\time}{t}
\safemath{\specfreq}{\theta}	
\newcommand{\transopgen}[1]{\opT_{#1}} 
\safemath{\transop}{\transopgen{\delay}}
\newcommand{\modopgen}[1]{\opM_{#1}}	
\safemath{\modop}{\modopgen{\shift}}
\newcommand{\dilopgen}[1]{\opD_{#1}}	
\safemath{\dilop}{\dilopgen{\scale}}
\safemath{\fram}{\setF}				
\safemath{\dfram}{\dual{\fram}}		
\safemath{\ufb}{B}					
\safemath{\lfb}{A}					
\safemath{\sop}{\hadj{\aop}}				
\safemath{\aop}{\opT}			
\safemath{\fop}{\opS}				
\safemath{\daop}{\tilde\opT}			
\safemath{\dsop}{\hadj{\tilde\opT}}				

\safemath{\ifop}{\inv{\fop}}			
\safemath{\rifop}{\fop^{-1/2}}			
\safemath{\transeq}{\setT}			
\safemath{\nfun}{\Phi}				
\safemath{\funvec}{\vecf}			
\safemath{\altfunvec}{\vecg}

\safemath{\samplespace}{\Omega}
\safemath{\probspace}{(\samplespace,\sfield,\Prob)}	
\safemath{\ccoef}{\rho}			

\safemath{\infstate}{\vecpi}				
\safemath{\typset}{\setA_{\epsilon}^{(N)}}	
\safemath{\expequal}{\doteq}				
\safemath{\code}{C}						
\safemath{\dstringset}{\setD^{\star}}		
\safemath{\cwlength}{l}					
\safemath{\elength}{L}					
\safemath{\extension}{C^{\star}}			
\safemath{\approaches}{\rightarrow}		
\safemath{\evnt}{\setA}					
\safemath{\altevnt}{\setB}					
\safemath{\rv}{\rndx}					
\safemath{\altrv}{\rndy}					
\safemath{\complexrv}{\rndu}					
\safemath{\altcrv}{\rndv}				
\safemath{\rvec}{\rvecx}					
\safemath{\altrvec}{\rvecy}				
\safemath{\crvec}{\rvecu}				
\safemath{\altcrvec}{\rvecv}				
\safemath{\variance}{\sigma^{2}}			
\safemath{\map}{T}						
\safemath{\jacobian}{\matJ}					
\safemath{\wvec}{\rvecw}					
\safemath{\wrv}{\rndw}					
\safemath{\orthmat}{\matQ}				
\safemath{\evmat}{\matLambda}			
\safemath{\identity}{\matidentity}		
\safemath{\innovec}{\vecv}				
\safemath{\convas}{\xrightarrow{\text{a.s.}}}	
\safemath{\convr}{\xrightarrow{\text{r}}}	
\safemath{\convp}{\xrightarrow{\text{P}}}	
\safemath{\convd}{\xrightarrow{\text{D}}}	
\safemath{\ltis}{\opL}				
\safemath{\ir}{h}					
\safemath{\tf}{\MakeUppercase{\ir}}	

\newcommand*{\fancyrefparlabelprefix}{par}		
\newcommand*{\fancyrefremlabelprefix}{rem}		
\newcommand*{\fancyrefchalabelprefix}{cha}		
\newcommand*{\fancyrefapplabelprefix}{app}		
\newcommand*{\fancyrefthmlabelprefix}{thm}		
\newcommand*{\fancyreflemlabelprefix}{lem}		
\newcommand*{\fancyrefcorlabelprefix}{cor}		
\newcommand*{\fancyrefdeflabelprefix}{def}		
\newcommand*{\fancyrefproplabelprefix}{prop}		

\frefformat{vario}{\fancyrefparlabelprefix}{Part~#1}
\frefformat{vario}{\fancyrefremlabelprefix}{Remark~#1}
\frefformat{vario}{\fancyrefchalabelprefix}{Chapter~#1}
\frefformat{vario}{\fancyrefseclabelprefix}{Section~#1}
\frefformat{vario}{\fancyrefthmlabelprefix}{Theorem~#1}
\frefformat{vario}{\fancyreflemlabelprefix}{Lemma~#1}
\frefformat{vario}{\fancyrefcorlabelprefix}{Corollary~#1}
\frefformat{vario}{\fancyrefdeflabelprefix}{Definition~#1}
\frefformat{vario}{\fancyreffiglabelprefix}{Figure~#1}
\frefformat{vario}{\fancyrefapplabelprefix}{Appendix~#1}
\frefformat{vario}{\fancyrefeqlabelprefix}{(#1)}
\frefformat{vario}{\fancyrefproplabelprefix}{Property~#1}

\theoremstyle{plain}
\newtheorem{thm}{Theorem}
\newtheorem{lem}{Lemma}

\newtheorem{dfn}{Definition}

\newtheorem{prp}{Proposition}
\theoremstyle{remark}
\newtheorem{rem}{Remark}
\newtheorem{exa}{Example}

\newcommand{\Proba}{\matP}

\allowdisplaybreaks[2]

\renewcommand{\dots}{\!...}

\begin{document}
%
\title{Almost Lossless Analog Compression\\ without Phase Information\\[-1mm]}

\author{\IEEEauthorblockN{Erwin Riegler}
\IEEEauthorblockA{ETH Zurich, 8092 Zurich, Switzerland\\
eriegler@nari.ee.ethz.ch\\[-3mm] } \and \IEEEauthorblockN{Georg Taub\"{o}ck}\thanks{This work was supported in part by the WWTF under grant VRG 12-009 and by the FWF under grant Y 551-N13.}
\IEEEauthorblockA{Austrian Academy of Sciences, 1040 Vienna, Austria\\
georg.tauboeck@oeaw.ac.at\\[-3mm] }}

\maketitle

\begin{abstract}
We propose an information-theoretic framework for phase retrieval. Specifically, we consider the problem of recovering an unknown vector $\vecx\in\reals^n$ up to an overall sign factor from $m=\lfloor Rn\rfloor$ phaseless measurements with compression rate $R$ 
and derive a general achievability bound for $R$.
Surprisingly, it turns out that this bound on the compression rate is the same as the one for almost lossless analog compression obtained by Wu and Verd\'u (2010):
Phaseless linear
measurements are ``as good'' as linear measurements with full phase
information in the sense that ignoring the sign  of $m$ measurements  only leaves us with an
ambiguity with respect to  an overall sign factor of $\vecx$.
\end{abstract}


%
\IEEEpeerreviewmaketitle


\section{Introduction}
In many different areas of science, physical limitations make it impossible
to measure the sign (phase in the complex case) of a signal but obtaining amplitudes is
relatively easy. Well known examples are X-ray crystallography, astronomy, or
diffraction imaging \cite{mill90,fienup93,qui10}.
The problem of retrieving a signal up to a global sign (phase in the complex case)
from intensity measurements is often referred to as \emph{phase
retrieval}. More formally, let $\reals_\sim^n$ be the set of equivalence classes
 $[\vecx]=\{\vecx\}\cup\{-\vecx\}$ with $\vecx\in\reals^n$. Phase retrieval  is the problem of recovering
$[\vecx]\in\reals_\sim^n$ from  $m$ phaseless measurements of the
form\footnote{For a vector $\vecu\in\reals^k$, we define the element-wise
absolute value operation as $|\vecu|=\tp{(|u_1|,\dots,|u_k|)}$.}
$\vecy=|\matA\vecx| \in \reals^m$
with measurement matrix $\matA\in\reals^{m\times n}$.

It is by no means clear how large $m$ has to be to allow for recovery of
$[\vecx]\in\reals^n_\sim$ from $m$ phaseless measurements. Thus
from the very beginning, there have been a number of works regarding
recovery  conditions for this problem in the context of specific
applications \cite{bruck79}. More recently, this question has been studied in
more abstract terms, asking for the smallest number  $m$ of phaseless measurements that
is required to make the mapping
$[\vecx]\mapsto |\matA\vecx|$ injective  without imposing structural
assumptions on $\matA$. In \cite{balan06}, the authors showed that
at least $2n-1$ such measurements are necessary and generically sufficient to
guarantee injectivity. 
Furthermore, it was shown that semidefinite programming can be used to
recover $[\vecx]$ if $\rmatA$ is random with i.i.d. Gaussian entries or with
i.i.d. rows that are uniformly distributed on a sphere, as long as $m \geq
c_0 n$ for a sufficiently large constant $c_0$ \cite{cand12a}. Other
phase retrieval methods for which theoretical performance guarantees are
available can be found, e.g., in \cite{cand13a,wald12,netr12,cand14a}.


Recently, there has been also interest in \emph{sparse phase retrieval},
where the number $s$ of nonzero coefficients of the vector $\vecx$ is much
smaller than  $n$. This a-priori knowledge about $\vecx$
can be used to reduce the number of measurements significantly. For instance,
$\mathcal{O}(s \log(n/s))$ measurements were shown to be sufficient for
stable sparse phase retrieval \cite{eldar14}.  If the rows of the measurement
matrix $\matA$ are a generic choice of vectors in $\reals^n$, injectivity  of
the mapping $[\vecx]\mapsto |\matA\vecx|$ is guaranteed provided that
$m\geq 2s$ \cite{akta14}.


\emph{Contributions:}
Following  the approach introduced for compressed sensing \cite{wuve10} and signal separation \cite{stribo13}
 problems,  we formulate phase retrieval as an analog source coding problem.
Assuming that the unknown vector $\rvecx$ is random with a certain distribution,
we derive asymptotic recovery results for $[\rvecx]$. Our results hold for Lebesgue almost all
(a.a.) measurement matrices $\matA$. However, our results are in terms of probability of error (with respect to the distribution
of $\rvecx$) 
and hence
do not provide worst-case guarantees.
Specifically, we study the asymptotic setting $n \rightarrow \infty$ where the vector $\vecx$ is a realization of a
random process; for each $n$, we let $m=\lfloor Rn \rfloor$ for a parameter
$R$, which we denote \emph{compression rate}.  
In Theorem \ref{thm1} we show that we can recover $[\rvecx]$ from $m$
phaseless measurements with arbitrarily small probability of error for a.a.
measurement matrices $\matA$, provided that $n$ is sufficiently large and the
compression rate $R$ is larger than the (lower) Minkowski dimension
compression rate (see Definition \ref{dfndimrate}) of  $\rvecx$. It is
remarkable that the obtained result is identical to the corresponding result
in compressive sensing \cite{wuve10} where $\vecy=\matA\vecx$, so that we can
conclude that \emph{in terms of achievability results, phaseless linear
measurements are ``as good'' as linear measurements with full phase
information:} Ignoring the sign  of $m$ measurements  only leaves us with an
ambiguity with respect to  an overall sign factor of $\vecx$.


\emph{Notation:}
Roman letters $\matA,\matB,\ldots$ and $\veca,\vecb,\ldots$
designate deterministic matrices and vectors, respectively.
Boldface letters $\rmatA,\rmatB,\ldots$ and $\rveca,\rvecb,\ldots$ denote random matrices and
random vectors, respectively.
For the distribution of a random matrix $\rmatA$ and a random vector $\rveca$, we write $\mu_\rmatA$ and  $\mu_\rveca$, respectively.
The $i$th component of the vector $\vecu$ (random vector $\rvecu$) is $\vecuc_i$ ($\rvecuc_i$).
The superscript  $\tp{}$ stands for transposition. For a matrix $\matA$, $\tr(\matA)$ denotes its trace.  
The identity matrix of suitable size is denoted by $\matI$.
For a vector $\vecu$, we write $\|\vecu\|=\sqrt{\tp{\vecu}\vecu}$ for its Euclidean norm.
For the Euclidean space $(\reals^k,\|\cdot\|)$, we denote the open ball of
radius $r$ centered at $\vecu\in \reals^k$ by $\setB_k(\vecu,r)$, $V(k,r)$
stands for its volume.
The Borel sigma algebra on $\reals$ is denoted by $\colB_\reals$.
We write $\reals_\geq$ for the set of nonnegative real numbers with Borel sigma algebra $\colB_{\reals_\geq}$.
For $\vecu,\vecv\in\reals^k$, $\vecu\equivalent \vecv$ means that either $\vecu=\vecv$ or $\vecu=-\vecv$ and we write  for the corresponding equivalence classes $[\vecu]=\{\vecu\}\cup\{-\vecu\}$. 
For a set $\setS\subseteq\reals^k$, $\setS_\sim=\{[\vecu] \mid \vecu\in\setS\}$.
The indicator function on a set $\setU$ is denoted by $\ind{\setU}$.



\section{Main Results}
We start by formulating phase retrieval as a  source coding problem.

\begin{dfn} (Source vector)\label{dfnsource}
Let $(\rvecxc_i)_{i\in\naturals}$ be a stochastic process on
$(\reals^\naturals,\colB_\reals^{\otimes\naturals})$. Then, for
$n\in\naturals$, the source vector $\rvecx$ of length $n$ is given by
$\rvecx=\tp{(\rvecxc_1,\dots,\rvecxc_n)}\in\reals^n$.
\end{dfn}

\begin{dfn} (Code, achievable rate) \label{dfncode}
For $\rvecx$ as in Definition \ref{dfnsource} and $\varepsilon >0$, an
$(n,m)$ code consists of
\begin{enumerate}[(i)]
\item measurements $|\matA \cdot|:\reals^n\to \reals^m_\geq$;
\item a decoder $g: \reals^m_\geq\to \reals^n$ that is measurable with
    respect to $\colB_{\reals_\geq}^{\otimes m}$ and
    $\colB_\reals^{\otimes n}$.
\end{enumerate}
We call $R$ with $0< R\leq 1$ an $\varepsilon$-achievable rate if there
exists an $N(\varepsilon)\in\naturals$ and a sequence of $(n,\lfloor
Rn\rfloor)$ codes with decoders $g$ such that
\begin{align*}
\Proba[g(|\matA \rvecx|)\not\equivalent \rvecx ] \leq \varepsilon
\end{align*}
for all $n\geq N(\varepsilon)$.
\end{dfn}

Next, we introduce the Minkowski dimension compression rate for source
vectors. 

\begin{dfn} (Minkowski dimension)\label{dfndim}
Let $\setU$ be a nonempty bounded set in $\reals^n$. The lower Minkowski
dimension of $\setU$ is defined as
\begin{align*}
\underline{\dim}_\mathrm{B}(\setU)=\liminf_{\rho\to 0} \frac{\log N_\setU(\rho)}{\log \frac{1}{\rho}}
\end{align*}
and the upper Minkowski dimension of $\setU$ is defined as
\begin{align*}
\overline{\dim}_\mathrm{B}(\setU)=\limsup_{\rho\to 0} \frac{\log N_\setU(\rho)}{\log \frac{1}{\rho}}
\end{align*}
where $N_\setU(\rho)$ is the covering number of $\setU$ given by
\begin{align*}
N_\setU(\rho)&=\min\Big\{k \in\naturals\mid \setU\subseteq \bigcup_{i\in\{1,\dots,k\}} \setB_n(\vecu_i,\rho),\ \vecu_i\in \reals^n\Big\}.
\end{align*}
If $\underline{\dim}_\mathrm{B}(\setU)=\overline{\dim}_\mathrm{B}(\setU)$, we
write $\dim_\mathrm{B}(\setU)$.
\end{dfn}

\begin{dfn}(Minkowski dimension compression rate)\label{dfndimrate}
For $\rvecx$ from Definition \ref{dfnsource} and $\varepsilon >0$, we define
the lower Minkowski dimension compression rate as
\begin{align}
\underline{R}_\mathrm{B}(\varepsilon)&=\limsup_{n\to\infty} \underline{a}_n(\varepsilon),\quad\text{where} \nonumber \\
\underline{a}_n(\varepsilon)&=\inf\Big\{\frac{\underline{\dim}_\mathrm{B}(\setU)}{n}   \;\Big\vert\;  \setU \subset\reals^n,\ \Proba[\rvecx\in\setU]\ \geq 1-\varepsilon\Big\}.\nonumber 
\end{align}
and the upper Minkowski dimension compression rate as
\begin{align}
\overline{R}_\mathrm{B}(\varepsilon)&=\limsup_{n\to\infty} \overline{a}_n(\varepsilon),\quad\text{where} \nonumber \\
\overline{a}_n(\varepsilon)&=\inf\Big\{\frac{\overline{\dim}_\mathrm{B}(\setU)}{n}   \;\Big\vert\;  \setU \subset\reals^n,\ \Proba[\rvecx\in\setU]\ \geq 1-\varepsilon\Big\}.\nonumber 
\end{align}
The sets $\setU$ in the definitions for $\underline{a}_n(\varepsilon)$ and
$\overline{a}_n(\varepsilon)$ are assumed to be nonempty and bounded.
\end{dfn}

\begin{exa}\label{exa1}
The source vector $\rvecx$ from Definition \ref{dfnsource} has a
mixed discrete-continuous distribution if for each $n\in \naturals$ the
random variables $\rvecxc_i$, $i\in\{1,\dots,n\}$, are independent and
distributed according to
\begin{align*}
\mu_{\rvecxc_i}=(1-\lambda)\mu_{{d}}+\lambda\mu_{{c}},\quad i\in\{1,\dots, n\}
\end{align*}
where $0\leq \lambda\leq 1$ is the mixing parameter, $\mu_{{c}}$ is a
distribution on $(\reals,\setB_\reals)$, absolutely continuous with respect
to Lebesgue measure, and $\mu_{{d}}$ is a discrete distribution. Then,
\cite[Th. 15]{wuve10}
\begin{align}
\underline{R}_\mathrm{B}(\varepsilon)=\overline{R}_\mathrm{B}(\varepsilon)= \lambda,\quad 0<\varepsilon <1.\nonumber
\end{align}
\end{exa}

The following result states that every rate  $R>\underline{R}_\mathrm{B}(\varepsilon)$ is $\varepsilon$-achievable for Lebesgue a.a. matrices $\matA$. 
\begin{thm}\label{thm1}
Let $0<\varepsilon < 1$ and $\rvecx$ as in Definition \ref{dfnsource}.
Then, for Lebesgue a.a. matrices
$\matA\in\reals^{m\times n}$ with $m=\lfloor Rn\rfloor$,  $R$ is an
$\varepsilon$-achievable rate provided that $R>\underline{R}_\mathrm{B}(\varepsilon)$.
\end{thm}

\begin{proof}
Since $R>\underline{R}_\mathrm{B}(\varepsilon)$ and $m=\lfloor Rn\rfloor$, 
Definition \ref{dfndimrate} implies that there exists a sequence of nonempty
bounded sets $\setU_n\subseteq \reals^n$ and an $N(\varepsilon)\in\naturals$
such that
\begin{align}
\underline{\dim}_\mathrm{B}(\setU) &< m\label{eq:nk}\\
\Proba\big[\rvecx\in\setU\big]&\geq 1-\varepsilon \label{eq:ek}
\end{align}
for all $\setU=\setU_n$ with $n\geq N(\varepsilon)$. In the remainder of the
proof we assume that $n$ is sufficiently large for  \eqref{eq:nk} and
\eqref{eq:ek} to hold.
The claim now follows from Proposition \ref{pro1} below.
\end{proof}

\begin{prp}\label{pro1}
Let $\varepsilon \geq 0$, $\rvecx\in\reals^n$ a random vector, and 
$\setU\subseteq\reals^n$  a nonempty bounded set with
$\Proba[\rvecx\in\setU]\geq 1-\varepsilon$. Then, for Lebesgue
a.a. matrices $\matA\in\reals^{m\times n}$, there exists a decoder $g$ with
$\Proba[g(|\matA \rvecx|)\not\equivalent \rvecx ] \leq \varepsilon$ provided
that $\underline{\dim}_\mathrm{B}(\setU) < m$.
\end{prp}
\begin{proof}
See Section \ref{proofthm1}.
\end{proof}

\begin{rem}
By \cite[Sec. 3.2, Properties (i)--(iii)]{fa90}, the lower Minkowski
dimension of any bounded nonempty subset in $\reals^n$ containing only
vectors
with no more than  $s$  nonzero entries  is at most $s$.
Therefore,  Proposition  \ref{pro1} implies that any $s$-sparse random vector
$\rvecx\in\reals^n$ can be recovered with arbitrarily  small probability of
error (by increasing the size of the set $\setU$ in Proposition \ref{pro1}),
provided that $m>s$.  This result holds for an arbitrary distribution of
$\rvecx$ and a.a. matrices  $\matA\in\reals^{m\times n}$. The best known
recovery threshold for deterministic $s$-sparse vectors is $m\geq 2s$
\cite{akta14}.

%
%
\end{rem}
\begin{rem}\label{remmc}
It is worth noting that formally phase retrieval can be formulated as a
matrix completion problem with measurements
$\vecyc_i^2=\tr(\veca_i\tp{\veca_i}\vecx\tp{\vecx})$ using  rank-one
measurement matrices $\matA_i=\veca_i\tp{\veca_i}$, $i=1,\dots,m$. However,
compared to the rank-one measurement matrices used in the matrix completion
problem \cite{cazh14,ristbo15}, the matrices $\veca_i\tp{\veca_i}$ are
symmetric. This complicates the proof of Proposition \ref{pro1} 
significantly and forces us to develop a novel
concentration of measure result (Lemma \ref{lemcomindep}). On the other hand, in phase retrieval we are interested in
recovering symmetric rank-one matrices $\vecx\tp{\vecx}$ (which is equivalent
to the recovery of $[\vecx]$), whereas matrix completion deals with the
recovery of arbitrary low-rank matrices.
\end{rem}
In the mixed discrete-continuous case we can strengthen the result of Theorem \ref{thm1} through the following lemma.
\begin{lem}
Let $0<\varepsilon < 1$ and $\rvecx$ be distributed according to the mixed
discrete-continuous distribution in Example \ref{exa1} with mixing parameter
$\lambda$. Then, for Lebesgue a.a. matrices $\matA\in\reals^{m\times n}$ with
$m=\lfloor Rn\rfloor$, $R$ is $\varepsilon$-achievable provided that
$R>\lambda$. Moreover, $R\geq \lambda$ is also a necessary condition for $R$
being $\varepsilon$-achievable.
\end{lem}
\begin{proof} Achievability: Follows from Theorem \ref{thm1} and Example \ref{exa1}.
Converse: Suppose that a rate $R<\lambda$ is $\varepsilon$-achievable for some
$\varepsilon$ with $0<\varepsilon<1$. This implies that there exists a set $\setK\subseteq\reals^n$
and a matrix $\matA\in\reals^{m\times n}$ with $m=\lfloor Rn\rfloor$ such that
\begin{enumerate}[(a)]
\item
$\Pr[\rvecx\in\setK]\geq 1-\varepsilon$;
\item
$|\matA \cdot|$ is one-to-one on $\setK_\sim$
\end{enumerate}
for $n$ sufficiently large. From $(b)$ it follows that there can be at most
one equivalence class $[\vecu]\in \setK_\sim$ with $\matA \vecu=\matA
(-\vecu)= 0$ 
(if there was more than one such
equivalence class then the mapping $|\matA \cdot|$ would not be one-to-one on
$\setK_\sim$).

Suppose first that there is no equivalence class $[\vecu]=\{\vecu,-\vecu\}\in \setK_\sim$ with $\matA \vecu=\matA (-\vecu)= 0$ and $\vecu\neq\veczero$. Then, (b) implies that $\matA$ is one-to-one on $\setK$ which, together with (a) and $R<\lambda$, leads to a  contradiction to the converse part of \cite[Thm. 6]{wuve10}.

Now suppose that there is an equivalence class $[\vecu]=\{\vecu,-\vecu\}\in
\setK_\sim$ with $\matA \vecu=\matA (-\vecu)= 0$ and $\vecu\neq\veczero$. Let
$\tilde R$ be such that $R < \tilde R < \lambda$ and set $\tilde m=\lfloor
\tilde R n\rfloor$. Then, $\tilde m > m$ for $n$ sufficiently large. Let
$\tilde\matA=\tp{(\tp{\matA},\vecu,\veczero,\dots,\veczero)}\in\reals^{\tilde
m\times n}$. Then, (b) implies that $\tilde\matA$ is one-to-one on $\setK$
which, together with (a) and $\tilde R<\lambda$, leads to a  contradiction to
the converse part of \cite[Thm. 6]{wuve10}.
\end{proof}



\section{Proof of Proposition \ref{pro1}}\label{proofthm1}
Let
\begin{align*}
\setF(\vecy)
&=\mleft\{\vecu\in\reals^n \big|\vecu\in\setU, |\matA \vecu|=\vecy\mright\}\\
&\phantom{=}\cup\mleft\{\vecu \in\reals^n\big|-\vecu\in\setU, |\matA \vecu|=\vecy\mright\},\quad \vecy\in\reals^m_\geq.
\end{align*}
For a vector $\vecu\in\setF(\vecy)\setminus\{\veczero\}$, let $\bar \vecuc$
denote the first nonzero component of $\vecu$. We then define the reduced set
\begin{align*}
\bar\setF(\vecy)=\mleft\{\vecu\in\setF(\vecy)\setminus\{\veczero\}\big| \bar\vecuc=|\bar\vecuc|\mright\}\cup(\setF(\vecy)\cap\{\veczero\}),\quad \!\!\vecy\in\reals^m_\geq.
\end{align*}
We define the decoder $g: \reals^m_\geq\to \reals^n$ by
\begin{align*}
g(\vecy)=
\begin{cases}
\vecu,&  \text{if}\ \bar\setF(\vecy) = \{\vecu\}\\
\vece, & \text{else}
\end{cases}
\end{align*}
where $\vece$ is some fixed vector in the complement of $\setU$ (used
to declare a decoding error). Then, we have
\begin{align}
&\Proba\big[g(|\matA\rvecx|)\not\equivalent\rvecx\big]\nonumber\\
&=\Proba\big[g(|\matA\rvecx|)\not\equivalent\rvecx,\rvecx\in\setU\big]+\Proba\big[g(|\matA\rvecx|)\not\equivalent\rvecx,\rvecx\notin\setU\big]\nonumber\\
&\leq\Proba\big[g(|\matA\rvecx|)\not\equivalent\rvecx,\rvecx\in\setU\big]+\varepsilon\nonumber\\
&=\Proba\mleft[\exists\vecu\in\setU \big| \vecu\not\equivalent\rvecx, |\matA\vecu|=|\matA\rvecx|, \rvecx\in\setU\mright]+\varepsilon \label{eq:errorbound1}
\end{align}
where  \eqref{eq:errorbound1} follows from the definition of the decoder. Fix an arbitrary $r>0$. Suppose that we can
show that
\begin{align}\label{eq:havetoshow}
P(\vecx) &= \Proba\big[\exists \vecu\in\setU\ \text{with}\ \vecu\not\equivalent\vecx, |\rmatA\vecu|=|\rmatA\vecx|\big]=0,\quad \vecx\in\setU
\end{align}
where $\rmatA\in\reals^{m\times n}$ has independent rows that are uniformly
distributed on $\setB_n(\veczero,r)$. Then,
\begin{align}
&\int\limits_{\setA(r)}\Proba\mleft[\exists\vecu\in\setU \big| \vecu\not\equivalent\rvecx, |\matA\vecu|=|\matA\rvecx|, \rvecx\in\setU\mright]\operatorname{d}\!\mu_\rmatA \nonumber\\
&=\int\limits_{\setU}
\Proba\big[\exists \vecu\in\setU\ \text{with}\ \vecu\not\equivalent\vecx, |\rmatA\vecu|=|\rmatA\vecx|\big]
\operatorname{d}\!\mu_\rvecx\nonumber\\
&=0\label{eq:argumentFubini}
\end{align}
where we used Fubini's Theorem and set $\setA(r)=\setB_n(\veczero,r)\times\dots\times \setB_n(\veczero,r)$.
Since $r$ is
arbitrary, \eqref{eq:argumentFubini} implies that
\begin{align}\label{eq:zeroaa}
\Proba\mleft[\exists\vecu\in\setU \big| \vecu\not\equivalent\rvecx, |\matA\vecu|=|\matA\rvecx|, \rvecx\in\setU\mright]=0
\end{align}
for Lebesgue a.a. matrices $\matA$. Hence, combining
\eqref{eq:errorbound1} and \eqref{eq:zeroaa}  proves the Proposition provided
that we can show that \eqref{eq:havetoshow} holds, which is done in Section
\ref{proofhavetoshow}.



\section{Proof of \eqref{eq:havetoshow}}\label{proofhavetoshow}
Suppose first that $\vecx=\veczero$. Then, $P(\vecx)=0$ if and only if
\begin{align}\label{eq:casexzero}
\Proba\big[\exists \vecu\in\setU\setminus\{\veczero\}\ \text{with}\ \rmatA\vecu=\veczero]=0.
\end{align}
Since $\underline{\dim}_\mathrm{B}\big(\setU\big)<m$, \eqref{eq:casexzero} follows from  \cite[Prop. 1]{stribo13}.
Therefore, we can assume in what follows that $\vecx\neq\veczero$.

We can upper-bound $P(\vecx)\leq P_1(\vecx)+P_2(\vecx)$
with
\begin{align*}
P_i(\vecx)=\Proba\big[\exists \vecu\in\setU_i(\vecx)\ \text{with}\ |\rmatA\vecu|=|\rmatA\vecx|\big],\quad i\in\{1,2\}
\end{align*}
where we defined
\begin{align}
\setU_1(\vecx)&=\{\vecu\in\setU \mid \rank (\vecx,\vecu) =2\}\nonumber\\
\setU_2(\vecx)&=\{\vecu\in\setU \mid \rank (\vecx,\vecu) =1\}\setminus \{\vecu\in\setU| \vecu\equivalent\vecx\}.\nonumber
\end{align}
We have to show that $P_i(\vecx)=0$ for  $i\in\{1,2\}$.
First, we establish  $P_2(\vecx)=0$. We have (recall that $\vecx\neq \veczero$) 
\begin{align}
&P_2(\vecx)\nonumber\\
&=
\Proba\big[\exists \vecu\in\setU\ \text{with}\ \rank (\vecx,\vecu) =1, \vecu\not\sim\vecx,  |\rmatA\vecu|=|\rmatA\vecx|]\nonumber\\
&=\Proba\big[\rmatA\vecx=\veczero]\nonumber\\
&=0\nonumber
\end{align}
where we used \cite[Prop. 1]{stribo13} together with $\underline{\dim}_\mathrm{B}\big(\{\vecx\}\big)=0$ in the last step. It remains to show that  $P_1(\vecx)=0$.
To this end, we first present an auxiliary lemma.
\begin{lem}\label{lemprobtri}
Let $r>0$, $\emptyset \neq \setS\subseteq \setB_n(\veczero,L)$, $\rho
>0$, $\vecx\in \setB_n(\veczero,L)$, and $\rmatA\in\reals^{m\times n}$ with
independent rows that are uniformly distributed on $\setB_n(\veczero,r)$.
Then, there exist $\vecs_l(\rho)\in \setS$,  $l=1,\dots,N_\setS(\rho)$ with $N_\setS(\rho)$ being the
covering number of $\setS$, such that
%
%
\begin{align}\label{eq:probtri}
&\Proba\big[\exists \vecu\in\setS\ \text{with}\ \big\||\rmatA\vecu|-|\rmatA\vecx|\big\|\leq\rho\big]\nonumber\\
&\leq
\sum_{l=1}^{N_\setS(\rho)}\Proba\mleft[\big||\tp{\rveca}\vecs_l(\rho)|^2 - |\tp{\rveca}\vecx|^2\big|\leq 2Lr(2r+1)\rho\mright]^m
\end{align}
where $\rveca$ is uniformly distributed on $\setB_n(\veczero,r)$.
\end{lem}
\begin{proof}
Let $\setS \subseteq \bigcup_{l\in\{1,\dots,N_\setS(\rho)\}}
\setB_n(\vecv_l(\rho),\rho)$, $\vecv_l(\rho)\in
\reals^n$, be a \emph{minimal} covering of $\setS$ according to the
definition of the covering number, cf. Definition \ref{dfndim}. Then, there
exist $\vecs_l(\rho)\in
\setS\cap\setB_n(\vecv_l(\rho),\rho)$ for all
$l=1,\dots,N(\rho)$. Hence, the balls
$\setB_n(\vecs_l(\rho),2\rho)$ cover the set $\setS$ and have
centers in $\setS$. We can upper bound the lhs in  \eqref{eq:probtri} by
\begin{align}
&\Proba\big[\exists \vecu\in\setS\ \text{with}\ \big\||\rmatA\vecu|-|\rmatA\vecx|\big\|\leq\rho\big]\nonumber\\
&\!\leq \!\!\sum_{l=1}^{N_\setS(\rho)}\!\! \Proba\big[\exists \vecu\in\setS\cap\setB_n(\vecs_l(\rho),2\rho)\ \text{with}\ \big\||\rmatA\vecu|-|\rmatA\vecx|\big\|\leq\rho\big]\nonumber\\
&\!\leq \!\!\sum_{l=1}^{N_\setS(\rho)}\!\! \Proba\big[\exists \vecu\in\setS\cap\setB_n(\vecs_l(\rho),2\rho)\ \text{with}\
\big||\tp{\rveca}\vecu| - |\tp{\rveca}\vecx|\big|\leq\rho\big]^m\label{eq:Pr2}
\end{align}
where $\eqref{eq:Pr2}$ follows from the fact that the rows of $\rmatA$ are
independent and uniformly distributed on $\setB_n(\veczero,r)$. Using the
triangle inequality we obtain
\begin{align}\label{eq:tri1}
\big||\tp{\rveca}\vecs_l(\rho)| - |\tp{\rveca}\vecx|\big|
&\leq
\big||\tp{\rveca}\vecx| - |\tp{\rveca}\vecu|\big|+\big||\tp{\rveca}\vecu| - |\tp{\rveca}\vecs_l(\rho)|\big|.
\end{align}
The second term on the rhs
of \eqref{eq:tri1} can be further upper bounded by
\begin{align}
\big||\tp{\rveca}\vecu| - |\tp{\rveca}\vecs_l(\rho)|\big|
&\leq \big|\tp{\rveca}(\vecu -\vecs_l(\rho))\big|\nonumber\\
&\leq \|\rveca\|\|\vecu -\vecs_l(\rho)\|\nonumber\\
&\leq 2r\rho\label{eq:tri2}
\end{align}
where \eqref{eq:tri2} follows from  $\vecu\in
\setB_n(\vecs_l(\rho),2\rho)$ and
 $\rveca\in \setB_n(\veczero,r)$.
Combining \eqref{eq:tri1} and \eqref{eq:tri2} gives
\begin{align}\label{eq:tri3}
\big||\tp{\rveca}\vecx| - |\tp{\rveca}\vecu|\big|\geq
\big||\tp{\rveca}\vecs_l(\rho)| - |\tp{\rveca}\vecx|\big|
-2r\rho.
\end{align}
Using \eqref{eq:tri3} in \eqref{eq:Pr2} yields
\begin{align}
&\Proba\big[\exists \vecu\in\setS\ \text{with}\ \big\||\rmatA\vecu|-|\rmatA\vecx|\big\|\leq\rho\big]\nonumber\\
&\leq \sum_{l=1}^{N_\setS(\rho)}\Proba\mleft[\big||\tp{\rveca}\vecs_l(\rho)| - |\tp{\rveca}\vecx|\big|\leq  (2r+1)\rho\mright]^m\nonumber\\
&\leq
\sum_{l=1}^{N_\setS(\rho)}\Proba\mleft[\big||\tp{\rveca}\vecs_l(\rho)|^2 - |\tp{\rveca}\vecx|^2\big|\leq  2Lr(2r+1)\rho\mright]^m
\label{eq:Pr3}
\end{align}
where $\eqref{eq:Pr3}$ follows from $\big||\tp{\rveca}\vecs_l(\rho)|^2
- |\tp{\rveca}\vecx|^2\big|=\big|(|\tp{\rveca}\vecs_l(\rho)| +
|\tp{\rveca}\vecx|)(|\tp{\rveca}\vecs_l(\rho)| -
|\tp{\rveca}\vecx|)\big|\leq 2Lr\big||\tp{\rveca}\vecs_l(\rho)| -
|\tp{\rveca}\vecx|\big|$.
\end{proof}

We now continue with the proof of $P_1(\vecx)=0$. Since $\setU$ is a bounded set, there exists an $L\in\reals$ such that 
\begin{align}\label{eq:boundedU}
\|\vecu\|\leq L,\quad \vecu\in\setU. 
\end{align}
We define the sets $\setT_{j}(\vecx)$ by


\begin{align*}
\setT_j(\vecx)&=\mleft\{\vecu\in\setU_1(\vecx)\Big| \sqrt{\|\vecu\|^2\|\vecx\|^2-|\tp{\vecu}\vecx|^2}>\frac{1}{j}\mright\},\quad j\in\naturals.
\end{align*}
Since
\begin{align*}
P_1(\vecx)\leq \sum_{j\in\naturals}\Proba\big[\exists \vecu\in\setT_{j}(\vecx)\ \text{with}\ |\rmatA\vecu|=|\rmatA\vecx|\big]
\end{align*}
it is sufficient to prove that
\begin{align}
P_1^{(j)}(\vecx)=\Proba\big[\exists \vecu\in\setT_{j}(\vecx)\ \text{with}\ |\rmatA\vecu|=|\rmatA\vecx|\big]&=0\nonumber
\end{align}
for all $j\in\naturals$.
%
Suppose, by contradiction, that there exists a $j\in\naturals$  such that $P_1^{(j)}(\vecx)>0$.  Then,
\begin{align}
\liminf_{\rho\to 0} \frac{\log P_1^{(j)}(\vecx)}{\log \frac{1}{\rho}}
&=0.\label{eq:contr}
\end{align}
Furthermore, $\setT_{j}(\vecx)
\neq \emptyset$ and by \cite[Sec. 3.2, Property (ii)]{fa90} (recall that
$\setT_{j}(\vecx)\subseteq \setU_1(\vecx)\subseteq\setU$ ) we get
\begin{align}\label{eq:dimT}
\underline{\dim}_\mathrm{B}\big(\setT_{j}(\vecx)\big)<m.
\end{align}
We have
\begin{align}
&
\liminf_{\rho\to 0} \frac{\log P_1^{(j)}(\vecx)}{\log \frac{1}{\rho}} \nonumber\\
&
=\liminf_{\rho\to 0} \frac{\log\Proba\big[\exists \vecu\in\setT_{j}(\vecx)\ \text{with}\ |\rmatA\vecu|=|\rmatA\vecx|\big]}{\log \frac{1}{\rho}}\nonumber\\
&
\leq
\liminf_{\rho\to 0}\nonumber\\
&
\ \ \ \frac{\log
\Big(\!
\sum_{l=1}^{N_{\!\setT_{j}(\vecx)}(\rho)}\!\!\Proba\mleft[\big||\tp{\rveca}\vecs^{(j)}_{l}(\rho , \vecx)|^2 \! - \!|\tp{\rveca}\vecx|^2\big|\leq\tilde\rho\mright]^m\Big)}{\log \frac{1}{\rho}} \label{eq:Pr4b}\\
&\leq
\liminf_{\rho\to 0}
\frac{\log
\Big(
{\tilde\rho}^m\sum_{l=1}^{N_{\!\setT_{j}(\vecx)}(\rho)}{f\big(\tilde\rho,r,\vecs^{(j)}_{l}(\rho , \vecx),\vecx\big)}^m   \Big)}{\log \frac{1}{\rho}} \label{eq:Pr4c}\\
&\leq
\liminf_{\rho\to 0}
\frac{\log \big({\tilde\rho}^m N_{\setT_{j}(\vecx)}(\rho) {\tilde f(\tilde\rho,r,L,j)}^m\big)}{\log \frac{1}{\rho}}\label{eq:Pr4d}\\
&=\underline{\dim}_\mathrm{B}(\setT_{j}(\vecx))-m + m \lim_{\rho\to 0}
\frac{\log \tilde f(\tilde\rho,r,L,j)}{\log \frac{1}{\rho}}\nonumber\\ 
&=\underline{\dim}_\mathrm{B}(\setT_{j}(\vecx))-m\nonumber\\
&<0\label{eq:Pr4e}
\end{align}
where in 
\eqref{eq:Pr4b} we applied Lemma \ref{lemprobtri} with
$\setS=\setT_{j}(\vecx)$ and set $\tilde\rho=2Lr(2r+1)\rho$, 
\eqref{eq:Pr4c} follows from Lemma \ref{lemcomindep} below with
$\vecu=\vecs^{(j)}_{l}(\rho , \vecx)$, $\vecv=\vecx$, and $\delta=\tilde\rho$
where $f$ is defined in \eqref{eq:fR}, in 
\eqref{eq:Pr4d} we used that 
\begin{align}
&f\big(\tilde\rho,r,\vecs^{(j)}_{l}(\rho , \vecx),\vecx\big)\nonumber\\
&\leq 
\tilde f(\tilde\rho,r,L,j)\nonumber\\
&=\frac{2(2r)^{n-2}j}{V(n,r)}\Big(1+\log\Big(2+\frac{8r^2L^2}{\tilde\rho}\Big)\Big),\quad l=1,\dots, N_{\!\setT_{j}(\vecx)}(\rho)\nonumber
\end{align}
which follows from \eqref{eq:boundedU} and the fact that $\vecs^{(j)}_{l}(\rho , \vecx) \in
\setT_{j}(\vecx)$, $l=1,\dots, N_{\setT_{j}(\vecx)}(\rho)$, and 
in 
\eqref{eq:Pr4e} we applied \eqref{eq:dimT}.  
But \eqref{eq:Pr4e}  is a contradiction to  \eqref{eq:contr}.  Therefore, $P_1^{(j)}(\vecx)=0$ for all $j\in\naturals$,  
which implies in turn that $P_1(\vecx)=0$  and concludes the proof of
\eqref{eq:havetoshow}.

\vspace*{0mm}

\section{Concentration of measure result}

\begin{lem}\label{lemcomindep}
Let $r>0$, $\rveca$ be uniformly distributed on $\setB_n(\veczero,r)$,
$\matC=\vecu\tp{\vecu}-\vecv\tp{\vecv}$ with linearly independent vectors
$\vecu,\vecv\in\reals^n$, and $\delta>0$. Then
\begin{align}\label{eq:bound}
\Proba\big[|\tp{\rveca}\matC\rveca|\leq\delta\big] \leq \delta f(\delta,r,\vecu,\vecv)
\end{align}
with
\begin{align}
&f(\delta,r,\vecu,\vecv)=\nonumber\\
&\frac{2(2r)^{n-2}\Big(1+\log\Big(2+\frac{2r^2\big(\|\vecu+\vecv\|\|\vecu-\vecv\|-\big|\|\vecu\|^2-\|\vecv\|^2\big|\big)}{\delta}\Big)\Big)}{\sqrt{\|\vecu\|^2\|\vecv\|^2-|\tp{\vecu}\vecv|^2}V(n,r)}
\label{eq:fR}
\end{align}
\end{lem}
\begin{proof} We have
\begin{align}
&\Proba\big[|\tp{\rveca}\matC\rveca|\leq\delta\big]\nonumber\\
&=\frac{1}{V(n,r)}\int\limits_{\setB_n(\veczero,r)} \ind{\mleft\{\veca\in\reals^n \big| |\tp{\veca}\matC\veca|<\delta\mright\}} \operatorname{d}\!\veca\nonumber\\
&=\frac{1}{V(n,r)}\int\limits_{\setB_n(\veczero,r)} \ind{\mleft\{\veca\in\reals^n \big| |\tp{\veca}\matW\matR
\matJ\tp{\matR}\tp{\matW}\veca|<\delta\mright\}} \operatorname{d}\!\veca\label{eq:com1b}\\
&=\frac{1}{V(n,r)}\int\limits_{\setB_n(\veczero,r)} \ind{\mleft\{\vecb\in\reals^n \big| |\tp{\vecc}\matR
\matJ\tp{\matR}\vecc |<\delta\mright\}} \operatorname{d}\!\vecb\label{eq:com1c}\\
&\leq \frac{(2r)^{n-2}}{V(n,r)}\int\limits_{\setB_2(\veczero,r)} \ind{\mleft\{\vecc\in\reals^2 \big| |\tp{\vecc}\matR
\matJ\tp{\matR}\vecc |<\delta\mright\}} \operatorname{d}\!\vecc \label{eq:com1}
\end{align}
where \eqref{eq:com1b} follows from Lemma \ref{lemQR} with $\matR$ and
$\matJ$ defined in \eqref{eq:JR} and $\matW$ defined in \eqref{eq:W} and
\eqref{eq:com1c} follows from changing variables to $\veca =  \bar\matW\vecb$
with $\bar\matW=(\matW, \matZ)\in\reals^{n\times n}$ where
$\matZ\in\reals^{n\times (n-2)}$ is chosen in such a way that
$\bar\matW\tp{\bar\matW}=\matI$ and $\vecc=\tp{(\veccc_1,\veccc_2)}$ with
$\veccc_1=\vecbc_1$ and $\veccc_2=\vecbc_2$.

The bound \eqref{eq:detRJR} on the determinant of the matrix
$\matR\matJ\tp{\matR}$ implies that one eigenvalue of $\matR\matJ\tp{\matR}$,
say $\lambda_1$, is positive and the other eigenvalue of
$\matR\matJ\tp{\matR}$, say $-\lambda_2$, is negative.
We can assume
without loss of generality that $\lambda_1 \geq \lambda_2$. Using the
eigendecomposition
$\matR\matJ\tp{\matR}=\matU\diag(\lambda_1,-\lambda_2) \tp{\matU}$,
where $\matU\in\reals^{2\times 2}$ with $\matU\tp{\matU}=\matI$, and changing
variables to $\vecc=\matU\vecd$, we can further upper bound  \eqref{eq:com1}
by
\begin{align}
&\frac{(2r)^{n-2}}{V(n,r)}
\int\limits_{\setB_2(\veczero,r)}
\ind{\mleft\{\vecc\in\reals^2 \big| |\tp{\vecc}\matR\matJ\tp{\matR}\vecc |<\delta\mright\}} \operatorname{d}\!\vecc\nonumber\\
&=
\frac{(2r)^{n-2}}{V(n,r)}
\int\limits_{\setB_2(\veczero,r)}
\ind{\mleft\{\vecd\in\reals^2 \big| \mleft|\lambda_1\vecdc_1^2- \lambda_2\vecdc_2^2\mright|<\delta\mright\}} \operatorname{d}\!\vecd \nonumber\\
&=
\frac{(2r)^{n-2}}{\sqrt{\lambda_1\lambda_2}V(n,r)}\int\limits_{\reals^2}
\ind{\mleft\{\vect\in\reals^2 \big|\frac{\vectc_1^2}{\lambda_1}+\frac{\vectc_2^2}{\lambda_2}\leq r^2\mright\}}\nonumber\\
&\phantom{=\frac{(2r)^{n-2}}{\sqrt{\lambda_1\lambda_2}V(n,r)}\int\limits_{\reals^2}}
\times\ind{\mleft\{\vect\in\reals^2 \big| \mleft|\vectc_1^2- \vectc_2^2\mright|<\delta \mright\}}
\operatorname{d}\!\vect\label{eq:com2a}\\
&\leq
\frac{(2r)^{n-2}}{\sqrt{\lambda_1\lambda_2}V(n,r)}\int\limits_{\reals^2}
\ind{\mleft\{\vect\in\reals^2 \big|\vectc_1^2\leq \lambda_1r^2, \vectc_2^2\leq \lambda_2r^2\mright\}}\nonumber\\
&\phantom{=\frac{(2r)^{n-2}}{\sqrt{\lambda_1\lambda_2}V(n,r)}\int\limits_{\reals^2}}
\times\ind{\mleft\{\vect\in\reals^2 \big| \mleft|\vectc_1^2- \vectc_2^2\mright|<\delta \mright\}}
\operatorname{d}\!\vect\label{eq:com3a}
\end{align}
where in \eqref{eq:com2a} we changed variables to $\vect=\diag(\sqrt{\lambda_1},\sqrt{\lambda_2})\vecd$.
The integral in \eqref{eq:com3a} measures the area that is inside the rectangle $\{\vect\mid \vectc_1^2\leq \lambda_1r^2, \vectc_2^2\leq \lambda_2r^2\}$ and the two hyperbolas $\{\vect\mid\vectc_1^2- \vectc_2^2=\pm\delta\}$ (see Figure \ref{picture}). The bound \eqref{eq:bound} can then be established by performing the following to steps:
\begin{enumerate}
\item deriving an upper bound on the integral in \eqref{eq:com3a}.
\item finding an expression of the eigenvalues of $\matR\matJ\tp{\matR}$ in terms of the vectors $\vecu$ and $\vecv$,
\end{enumerate}
which will be done next.
We have
\pgfplotsset{every axis/.append style={
                    axis x line=middle,    
                    axis y line=middle,    
                    axis line style={<->}, 
                    xlabel={$t_1$},          
                    ylabel={$t_2$},          
                    }}

\tikzset{>=stealth}

\begin{figure}
\begin{center}
\begin{tikzpicture}
    \begin{axis}[
            xmin=-5,xmax=5,
        ymin=-5,ymax=5]
        \addplot [red,thick,domain=-2:2] ({cosh(x)}, {sinh(x)});
        \addplot [red,thick,domain=-2:2] ({-cosh(x)}, {sinh(x)});
        \addplot [red,thick,domain=-2:2] ({sinh(x)},{cosh(x)});
        \addplot [red,thick,domain=-2:2] ({sinh(x)},{-cosh(x)});
        \addplot[red,dashed] expression {x};
        \addplot[red,dashed] expression {-x};
        \addplot [blue,thick,domain=-4:4] ({x}, {2});
        \addplot [blue,thick,domain=-4:4] ({x}, {-2});
        \addplot [blue,thick,domain=-2:2] ({-4}, {x});
        \addplot [blue,thick,domain=-2:2] ({4}, {x});
    \end{axis}
\end{tikzpicture}
\caption{Intersection of the rectangle
$\{\vect\mid \vectc_1^2\leq \lambda_1r^2, \vectc_2^2\leq \lambda_2r^2\}$
with the two hyperbolas
$\{\vect\mid\vectc_1^2- \vectc_2^2=\pm\delta\}$  for
$\delta=1$, $\lambda_1=16/r^2$, and $\lambda_2=4/r^2$.}\label{picture}
\end{center}
\end{figure}
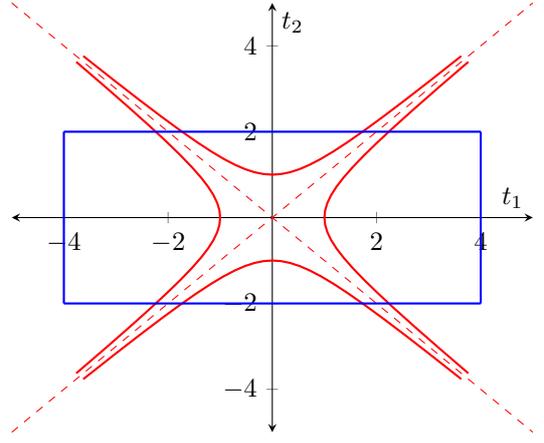

\begin{align}
&\frac{(2r)^{n-2}}{\sqrt{\lambda_1\lambda_2}V(n,r)}\int\limits_{\reals^2}
\ind{\mleft\{\vect\in\reals^2 \big|\vectc_1^2\leq \lambda_1r^2, \vectc_2^2\leq \lambda_2r^2\mright\}}\nonumber\\
&\phantom{=\frac{(2r)^{n-2}}{\sqrt{\lambda_1\lambda_2}V(n,r)}\int\limits_{\reals^2}}
\times\ind{\mleft\{\vect\in\reals^2 \big| \mleft|\vectc_1^2- \vectc_2^2\mright|<\delta \mright\}}
\operatorname{d}\!\vect\nonumber\\
&\leq
\frac{(2r)^{n-2}}{\sqrt{\lambda_1\lambda_2}V(n,r)}\int\limits_{\reals^2}
\ind{\mleft\{\vect\in\reals^2 \big|\vectc_1^2+\vectc_2^2\leq \delta+2\lambda_2r^2\mright\}}\nonumber\\
&\phantom{=\frac{(2r)^{n-2}}{\sqrt{\lambda_1\lambda_2}V(n,r)}\int\limits_{\reals^2}}
\times\ind{\mleft\{\vect\in\reals^2 \big| \mleft|\vectc_1^2- \vectc_2^2\mright|<\delta \mright\}}
\operatorname{d}\!\vect
\label{eq:com2b}\\
&=
\frac{(2r)^{n-2}}{\sqrt{\lambda_1\lambda_2}V(n,r)}\int\limits_{\opR^2}
\ind{\mleft\{\vecz\in\opR^2 \big|\veczc_1^2+\veczc_2^2\leq \delta+2\lambda_2r^2\mright\}} \nonumber\\
&\phantom{=\frac{(2r)^{n-2}}{\sqrt{\lambda_1\lambda_2}V(n,r)}\int\limits_{\reals^2}}
\times
\ind{\mleft\{\vecz\in\opR^2 \big| \mleft|\veczc_1\veczc_2\mright|<\frac{\delta}{2} \mright\}}
\operatorname{d}\!\vecz
\label{eq:com2c}\\
&\leq
\frac{(2r)^{n-2}}{\sqrt{\lambda_1\lambda_2}V(n,r)}\int\limits_{\opR^2}
\ind{\mleft\{\vecz\in\opR^2 \big|\veczc_1^2\leq \delta+2\lambda_2r^2,\, \veczc_2^2\leq \delta+2\lambda_2r^2\mright\}}\nonumber\\
&\phantom{=\frac{(2r)^{n-2}}{\sqrt{\lambda_1\lambda_2}V(n,r)}\int\limits_{\reals^2}}
\times
\ind{\mleft\{\vecz\in\opR^2 \big| \mleft|\veczc_1\veczc_2\mright|<\frac{\delta}{2} \mright\}}
\operatorname{d}\!\vecz
\nonumber\\
&=
\frac{4(2r)^{n-2}}{\sqrt{\lambda_1\lambda_2}V(n,r)}\int\limits_{\opR_\geq^2}
\ind{\mleft\{\vecz\in\opR^2 \big|\veczc_1\leq \sqrt{\delta+2\lambda_2r^2}\mright\}}\nonumber\\
&\phantom{=\frac{4(2r)^{n-2}}{\sqrt{\lambda_1\lambda_2}V(n,r)}\int\limits_{\opR_\geq^2}}
\times
\ind{\mleft\{\vecz\in\opR^2 \big|\veczc_2\leq \min\big(\sqrt{\delta+2\lambda_2r^2},\frac{\delta}{2\veczc_1}\big)\mright\}}
\operatorname{d}\!\vecz
\nonumber\\
&\leq
\frac{4(2r)^{n-2}}{\sqrt{\lambda_1\lambda_2}V(n,r)}\int\limits_{\opR_\geq^2}
\ind{\mleft\{\vecz\in\opR^2 \big|\veczc_1\leq \frac{\delta}{2\sqrt{\delta+2\lambda_2r^2}}\mright\}}\nonumber\\
&\phantom{=\frac{4(2r)^{n-2}}{\sqrt{\lambda_1\lambda_2}V(n,r)}\int\limits_{\opR_\geq^2}}
\times
\ind{\mleft\{\vecz\in\opR^2 \big|\veczc_2\leq  \sqrt{\delta+2\lambda_2r^2}\mright\}}
\operatorname{d}\!\vecz\nonumber\\
&\phantom{=}
+\frac{4(2r)^{n-2}}{\sqrt{\lambda_1\lambda_2}V(n,r)}\int\limits_{\opR_\geq^2}
\ind{\mleft\{\vecz\in\opR^2 \big| \frac{\delta}{2\sqrt{\delta+2\lambda_2r^2}}< \veczc_1\leq\sqrt{\delta+2\lambda_2r^2}\mright\}}\nonumber\\
&\phantom{=\frac{4(2r)^{n-2}}{\sqrt{\lambda_1\lambda_2}V(n,r)}\int\limits_{\opR_\geq^2}}
\times
\ind{\mleft\{\vecz\in\opR^2 \big|\veczc_2\leq \frac{\delta}{2\veczc_1}\mright\}}
\operatorname{d}\!\vecz
\nonumber\\
&=\frac{2\delta(2r)^{n-2}}{\sqrt{\lambda_1\lambda_2}V(n,r)}\Big(1+\log\Big(2+\frac{4\lambda_2r^2}{\delta}\Big) \Big)
\label{eq:com3R}
\end{align}

where in \eqref{eq:com2b} we used  that $\vect_2^2 \leq \lambda_2r^2$ and
$|\vectc_1^2- \vectc_2^2|<\delta$ imply  $\vectc_1^2 + \vectc_2^2 \leq
\delta+2\lambda_2r^2$, and in \eqref{eq:com2c} we applied the orthogonal transformation
$z_1=(1/\sqrt{2})(t_1+t_2)$, $z_2=(1/\sqrt{2})(t_1-t_2)$.
Combining  \eqref{eq:com1} with \eqref{eq:com3R} and using the expressions
\eqref{eq:detRJR} and  \eqref{eq:lminRJR} gives \eqref{eq:fR}.
\end{proof}

\section{Properties of certain rank two matrices}
%
%
%

\begin{lem}\label{lemQR}
Let $\vecu,\vecv\in\reals^n$ be linearly independent and
$\matC=\vecu\tp{\vecu}-\vecv\tp{\vecv}$. Then,
\begin{align}\label{eq:CQR}
\matC=
\matW
\matR
\matJ
\tp{\matR}
\tp{\matW}
\end{align}
with
\begin{align}\label{eq:JR}
\matJ=
\begin{pmatrix}
1&0\\
0&-1
\end{pmatrix},&&
\matR=
\begin{pmatrix}
\|\vecu\|&\frac{\tp{\vecu}\vecv}{\|\vecu\|}\\
0&\|\vecv-\frac{\tp{\vecu}\vecv}{\|\vecu\|^2}\vecu\|
\end{pmatrix}
\end{align}
and
\begin{align}\label{eq:W}
\matW=\begin{pmatrix}
\frac{\veca}{\|\veca\|},\frac{\vecb}{\|\vecb\|}
\end{pmatrix}
\end{align}
where the orthonormal vectors $\veca/\|\veca\|$ and $\vecb/\|\vecb\|$ are
defined by
\begin{align}
\veca&=\vecu\label{eq:veca}\\
\vecb&=\vecv-\frac{\tp{\vecu}\vecv}{\|\vecu\|^2}\veca.\label{eq:vecb}
\end{align}
Moreover,
\begin{align}
\det(\matR\matJ\tp{\matR})
&=|\tp{\vecu}\vecv|^2-\|\vecu\|^2\|\vecv\|^2<0\label{eq:detRJR}\\
\tr(\matR\matJ\tp{\matR}) &= \|\vecu\|^2-\|\vecv\|^2\label{eq:trRJR}\\
\sigma_2(\matR\matJ\tp{\matR})
&=\frac{1}{2}\|\vecu+\vecv\|\|\vecu-\vecv\| - \frac{1}{2}\big|\|\vecu\|^2-\|\vecv\|^2\big|
\label{eq:lminRJR}
\end{align}
where $\sigma_1(\matR\matJ\tp{\matR})\geq\sigma_2(\matR\matJ\tp{\matR})$ are the singular values of $\matR\matJ\tp{\matR}$.
\end{lem}
\begin{proof}
We can rewrite $\matC=
\matA\matJ\tp{\matA}$ with
$\matA=(\vecu,\vecv)$. 
Hence, to prove \eqref{eq:CQR}, it is sufficient to show that $\matA=\matW\matR$. 

Using the definitions of the vectors $\veca$ and $\vecb$ in \eqref{eq:veca}
and \eqref{eq:vecb}, we can rewrite
\begin{align*}
\matA
&=
\begin{pmatrix}
\veca,\frac{\tp{\vecu}\vecv}{\|\vecu\|^2}\veca+\vecb
\end{pmatrix}\\
&=
\begin{pmatrix}
\veca,\vecb
\end{pmatrix}
\begin{pmatrix}
1& \frac{\tp{\vecu}\vecv}{\|\vecu\|^2}\\
0&1
\end{pmatrix}\\
&=
\begin{pmatrix}
\frac{\veca}{\|\veca\|},\frac{\vecb}{\|\vecb\|}
\end{pmatrix}
\begin{pmatrix}
\|\vecu\|& \frac{\tp{\vecu}\vecv}{\|\vecu\|}\\
0&\|\vecv-\frac{\tp{\vecu}\vecv}{\|\vecu\|^2}\vecu\|
\end{pmatrix}\\
&=\matW\matR
\end{align*}
which proves \eqref{eq:CQR}.

The explicit form of the determinant in \eqref{eq:detRJR} follows from the
fact that
\begin{align}
\det(\matR\matJ\tp{\matR})
&=\det(\matR)\det(\matJ)\det(\tp{\matR})\nonumber\\
&=-|\det(\matR)|^2\nonumber\\
&=-\|\vecu\|^2\mleft\|\vecv-\frac{\tp{\vecu}\vecv}{\|\vecu\|^2}\vecu\mright\|^2\nonumber\\
&=|\tp{\vecu}\vecv|^2-\tp{\vecu}\vecu\tp{\vecv}\vecv\nonumber\\
&<0\label{eq:cs}
\end{align}
where \eqref{eq:cs} follows from the Cauchy-Schwarz inequality \cite[Sec.
0.6.3]{hojo90} and  $\vecu$ and $\vecv$ being  linearly
independent. The expression for the trace \eqref{eq:trRJR} follows from
$\tr(\matR\matJ\tp{\matR})=\tr(\matC)$. Finally,  \eqref{eq:lminRJR} follows
from
\begin{align}
\sigma_2(\matR\matJ\tp{\matR})
&=\frac{1}{2}\big(\sigma_1(\matR\matJ\tp{\matR})+\sigma_2(\matR\matJ\tp{\matR})\big)\nonumber\\
&\phantom{=}-\frac{1}{2}\big(\sigma_1(\matR\matJ\tp{\matR})-\sigma_2(\matR\matJ\tp{\matR})\big)\nonumber\\
&=\frac{1}{2}\sqrt{\tr(\matR\matJ\tp{\matR})^2-4\det(\matR\matJ\tp{\matR})}-\frac{1}{2}|\tr(\matR\matJ\tp{\matR})|\nonumber\\
&=\frac{1}{2}\|\vecu+\vecv\|\|\vecu-\vecv\|- \frac{1}{2}\big|\|\vecu\|^2-\|\vecv\|^2\big|.\nonumber
\end{align}
\end{proof}

\bibliographystyle{IEEEtran}
\bibliography{IEEEabrv,references,isit15a}

\end{document}